\documentclass[11pt]{article}

\usepackage[margin=1in]{geometry}
\usepackage[utf8]{inputenc}
\usepackage{amsmath, amssymb, amsthm}
\usepackage{mathrsfs}
\usepackage{cite}
\usepackage{hyperref}
\usepackage{authblk}
\usepackage[ruled, vlined, linesnumbered]{algorithm2e}
\usepackage{multirow}
\usepackage{graphicx}
\usepackage[labelformat=simple]{subcaption}

\usepackage{booktabs, dcolumn}
\usepackage{xcolor}
\usepackage{tikz}
\usetikzlibrary{quantikz}

\newcommand{\bs}[1]{\boldsymbol #1}
\newtheorem{definition}{Definition}
\newtheorem{theorem}{Theorem}

\newtheorem{proposition}{Proposition}
\newtheorem{problem}{Problem}
\DeclareMathOperator{\rank}{rank}

\title{Optimal Hadamard gate count for Clifford$+T$ synthesis of\\ Pauli rotations sequences}

\author[1,2]{Vivien Vandaele}
\author[1]{Simon Martiel}
\author[2]{Simon Perdrix}
\author[2]{Christophe Vuillot}
\affil[1]{Atos Quantum Lab, Les Clayes-sous-Bois, France}
\affil[2]{Université de Lorraine, CNRS, Inria, LORIA, F-54000 Nancy, France}

\date{}

\begin{document}
\maketitle

\begin{abstract}
The Clifford$+T$ gate set is commonly used to perform universal quantum computation.
In such setup the $T$ gate is typically much more expensive to implement in a fault-tolerant way than Clifford gates.
To improve the feasibility of fault-tolerant quantum computing it is then crucial to minimize the number of $T$ gates.
Many algorithms, yielding effective results, have been designed to address this problem.
It has been demonstrated that performing a pre-processing step consisting of reducing the number of Hadamard gates in the circuit can help to exploit the full potential of these algorithms and thereby lead to a substantial $T$-count reduction.
Moreover, minimizing the number of Hadamard gates also restrains the number of additional qubits and operations resulting from the gadgetization of Hadamard gates, a procedure used by some compilers to further reduce the number of $T$ gates.
In this work we tackle the Hadamard gate reduction problem, and propose an algorithm for synthesizing a sequence of $\pi/4$ Pauli rotations with a minimal number of Hadamard gates.
Based on this result, we present an algorithm which optimally minimizes the number of Hadamard gates lying between the first and the last $T$ gate of the circuit.
\end{abstract}

\section{Introduction}

Fault-tolerant quantum computing enables reliable and large-scale quantum computation at the cost of an important resource overhead when compared to an error-free model.
Much work has been put into quantum circuit optimization in order to reduce this additional cost and make fault-tolerant quantum computing more practical and scalable.
In particular, numerous algorithms have been designed to minimize the number of $T$ gates in a quantum circuit~\cite{amy2014polynomial, gosset2014algorithm, abdessaied2014quantum, amy2019t, nam2018automated, heyfron2018efficient, kissinger2020reducing, zhang2019optimizing, de2019techniques, de2020fast, munson2019and, mosca2021polynomial}.
This focus on $T$-count minimization is primarily due to the sizable amount of resources, in terms of time and number of qubits, generally required by fault-tolerance protocols, such as magic state distillation~\cite{bravyi2005universal}, to implement the $T$ gate.
In contrast, Clifford operations can typically be implemented at little expense in most common quantum error correcting codes via transversal operations, code deformation~\cite{bombin2009quantum} or lattice surgery~\cite{horsman2012surface}.
In such context, and considering the fact that the Clifford$+T$ gate set is approximatively universal, the $T$-count stands out as a key metric to minimize in order to make fault-tolerant quantum computing more efficient.
Moreover, minimizing the $T$-count is also crucial in the field of quantum circuits simulation as many simulators have a runtime that scales exponentially with respect to the number of $T$ gates~\cite{bravyi2016improved, bravyi2019simulation, qassim2021improved, kissinger2022simulating, kissinger2022classical}.

The problem of finding the optimal number of $T$ gates in a $\{\mathrm{CNOT}, S, T\}$ circuit composed of $n$ qubits has been well formalized for $\{\mathrm{CNOT}, S, T\}$ circuits by demonstrating its equivalence with the problem of finding a maximum likelihood decoder for the punctured Reed-Muller code of order $n - 4$ and length $2^n-1$~\cite{amy2019t}, which is tantamount to the third order symmetric tensor rank decomposition problem~\cite{seroussi1983maximum}.
In order to make use of this formalism in Clifford$+T$ circuits it is necessary to circumvent the Hadamard gates in some way; this can be achieved by applying one of the two following strategies.
The first method consists of extracting $\{\mathrm{CNOT}, S, T\}$ subcircuits and interposing them with layers of Hadamard gates~\cite{amy2014polynomial}.
Then an independent and Hadamard-free instance of the $T$-count minimization problem can be formulated for each $\{\mathrm{CNOT}, S, T\}$ subcircuit extracted.
The second strategy involves a measurement-based gadget which can substitute a Hadamard gate.
This Hadamard gadgetization procedure requires the following additional resources for each Hadamard gate gadgetized: an ancilla qubit, a CZ gate and a measurement~\cite{bremner2011classical}.

The number $T$ gates in a circuit containing $h$ Hadamard gates can be upper bounded by $\mathcal{O}(n^2h)$ or $\mathcal{O}((n+h)^2)$ in the case where all Hadamard gates are gadgetized~\cite{amy2019t}.
Hence, each Hadamard gate that must be circumvented, regardless of the strategy applied, for a lack of a good Hadamard gate optimization procedure is potentially the cause of missed opportunities for further $T$ gate reduction.
Therefore, a preliminary procedure consisting in reducing the number of Hadamard gates can result in an important $T$-count reduction, as demonstrated in Reference~\cite{abdessaied2014quantum}.
It has been shown that circumventing all Hadamard gates using the Hadamard gadgetization procedure is the strategy that leads to the best reduction in the number of $T$ gates~\cite{heyfron2018efficient}.
However, the main drawback of this method is the use of one additional qubit for each Hadamard gate gadgetized.
This is obviously an inconvenience if the number of qubits at disposal is limited, but can also be detrimental to the optimization process in two ways.
Firstly, as suggested in Reference~\cite{de2020fast}, it may become more difficult to find opportunities to reduce the $T$-count as the ratio between the number of qubits and the number of $T$ gates increases.
In addition, the runtime of a $T$-count optimizer can drastically increase as the number of qubits grows.
For all these reasons it is important to minimize the number of auxiliary qubits needed, which further motivates investigations into a pre-processing step optimizing the number Hadamard gates in the initial circuit.

We can mainly distinguish two strategies for the optimization of quantum circuits.
The first one is referred to as pattern matching and involves the detection of patterns of gates within the circuit to then substitute them by an equivalent, but nonetheless different, sequence of gates.
A series of transformation is therefore applied to the circuit, but its semantic is preserved at each step of the process.
This method has already been applied to the optimization of Hadamard gates by using rewriting rules that preserve or reduce the number of Hadamard gates within the circuit~\cite{abdessaied2014quantum, de2020fast}.
The second method is circuit re-synthesis which consists in extracting some parts of the circuit, representing them by higher level constructs and performing their synthesis to obtain an equivalent circuit.
This method has not yet been considered for the optimization of Hadamard gates, despite displaying excellent performances for other optimization problems such as $T$ gate reduction~\cite{zhang2019optimizing, heyfron2018efficient}.

In the case of circuit re-synthesis, a commonly used fact is that the operation performed by a given Clifford$+T$ circuit can be represented by a sequence of $\pi/4$ Pauli rotations followed by a final Clifford operator~\cite{gosset2014algorithm}.
A strategy for optimizing the number of Hadamard gates could then consist of synthesizing this sequence of $\pi/4$ Pauli rotations using as few Hadamard gates as possible.
In Section~\ref{sec:h_opt}, we present an algorithm that solves this problem optimally.
With the Hadamard gadgetization approach, a Hadamard gate needs to be gadgetized only if it comes after and precedes a $T$ gate in the circuit, we say that such Hadamard gates are internal Hadamard gates.
This leads to a more specific Hadamard gate reduction problem consisting in reducing the number of internal Hadamard gates within the circuit.
We tackle this problem in Section~\ref{sec:internal} by proposing an algorithm that synthesizes a sequence of Pauli rotations with a minimal number of internal Hadamard gates.
Section~\ref{sec:improving} presents alternative versions of our algorithms with lower complexities.
Benchmarks are then given in Section~\ref{sec:bench} to evaluate the performances and scalability of our algorithms on a library of reversible logic circuits and on large-scale quantum circuits.
Our algorithms are not working exclusively for the Clifford$+T$ gate set but can also be executed on any circuit composed of $\{X, \mathrm{CNOT}, S, H, R_Z\}$ gates.

\section{Preliminaries}

\subsection{Pauli rotations sequences}

The four Pauli matrices are defined as follows:
$$
I = \begin{pmatrix}
1 & 0\\
0 & 1
\end{pmatrix},\quad
X = \begin{pmatrix}
0 & 1\\
1 & 0
\end{pmatrix},\quad
Y = \begin{pmatrix}
0 & -i\\
i & 0
\end{pmatrix},\quad
Z = \begin{pmatrix}
1 & 0\\
0 & -1
\end{pmatrix}.
$$
Two Pauli matrices commute if they are equal or if one of them is the identity matrix $I$, otherwise they anticommute.
All tensor products of $n$ Pauli matrices, together with an overall phase of $\pm 1$ or $\pm i$, generate the Pauli group $\mathcal{P}_n$.
We define the subset $\mathcal{P}^*_n \subset \mathcal{P}_n$ as the set of Pauli operators which have an overall phase of $\pm 1$.
We will use $P_i$ to denote the $i$th Pauli matrix of a Pauli operator $P$, for instance if $P = Z \otimes X$ then $P_1 = Z$ and $P_2 = X$.
We say that a Pauli operator $P$ is diagonal if and only if $P_i \in \{I, Z\}$ for all $i$.
Two Pauli operators $P$ and $P'$ commute if there is an even number of indices $i$ such that $P_i$ anticommutes with $P'_i$, otherwise they anticommute.
Given a Pauli operator $P \in \mathcal{P}^*_n$ and an angle $\theta \in \mathbb{R}$, a Pauli rotation $R_P(\theta)$ is defined as follows:
$$R_P(\theta) = \exp(-i\theta P/2) = \cos(\theta/2)I - i\sin(\theta/2)P.$$
For example the $T$ gate is defined as a $\pi/4$ Pauli $Z$ rotation: 
$$T = R_Z(\pi/4)$$
Clifford gates can also be represented in terms of Pauli rotations, we will mostly make use of the CNOT, $S$ and $H$ gates defined as follows:
    \begin{align*}
        \mathrm{CNOT} &= R_{ZX}(\pi/2)R_{ZI}(-\pi/2)R_{IX}(-\pi/2), \\
        S &= R_Z(\pi/2), \\
        H &= R_{Z}(\pi/2) R_{X}(\pi/2) R_{Z}(\pi/2).
    \end{align*}
The Clifford group $\mathcal{C}_n$ is defined as the set of unitaries stabilizing $\mathcal{P}_n$:
$$\mathcal{C}_n = \{U \mid U^\dag PU \in \mathcal{P}_n,\, \forall P \in \mathcal{P}_n \}.$$
and is generated by the $\{\text{CNOT}, S, H\}$ gate set.
Note that for each pair of Pauli operators $P, P' \in \mathcal{P}_n \setminus \{I^{\otimes n}\}$ there exists a Clifford operator $U \in \mathcal{C}_n$ such that $P' = U^\dag P U$.
Unless indicated otherwise, the term Clifford circuit will refer to a circuit exclusively composed of gates from the set $\{X, \text{CNOT}, S, H\}$, the use of other Clifford gate set is discussed at the end of Section~\ref{sec:diag_opt}.

A Pauli operator $P \in \mathcal{P}^*_n$ can be encoded using $2n + 1$ bits: $2n$ bits for the $n$ Pauli matrices and $1$ bit for the sign~\cite{aaronson2004improved}.
In the following we will encode a Pauli operator $P \in \mathcal{P}^*_n$ with $2n$ bits and neglect its sign as it has no impact on the formulation of our problem; we will use the term Pauli product to designate a Pauli operator deprived of its sign.
Let $\mathcal{S} = \begin{bmatrix} \mathcal{Z} \\ \mathcal{X} \end{bmatrix}$ be a block matrix of size $2n \times m$ representing a sequence of $m$ Pauli products acting on $n$ qubits such that $\mathcal{Z}$ is the submatrix of $\mathcal{S}$ formed by its first $n$ rows and $\mathcal{X}$ is the submatrix of $\mathcal{S}$ formed by its last $n$ rows.
The value $({\mathcal{Z}}_{i,j}, \mathcal{X}_{i,j})$ represents the $i$th component of the $j$th Pauli product encoded by $\mathcal{S}$, such that the values $(0, 0), (0, 1), (1, 1)$ and $(1, 0)$ are corresponding to the Pauli matrices $I, X, Y$ and $Z$ respectively.
We use the notation $\mathcal{S}_{:,i}$ to refer to the column $i$ of the matrix $\mathcal{S}$, we will denote $P(\mathcal{S}_{:,i})$ the Pauli product encoded by $\mathcal{S}_{:,i}$, and we will say that $\mathcal{S}_{:,i}$ is diagonal if and only if $P(\mathcal{S}_{:,i})$ is diagonal and that $\mathcal{S}_{:,i}$ and $\mathcal{S}_{:,j}$ commute (or anticommute) if their associated Pauli products $P(\mathcal{S}_{:,i})$ and $P(\mathcal{S}_{:,j})$ commute (or anticommute).
Throughout the document we use the zero-based indexing for vectors and matrices and the initial element is termed the zeroth element, for instance the zeroth column of $\mathcal{S}$ is $\mathcal{S}_{:,0}$.
If all Pauli products encoded by $\mathcal{S}$ are conjugated by a Clifford gate $U \in \{\text{CNOT}, S, H\}$, then $\mathcal{S}$ can be updated to encode the Pauli products $U^\dag P(\mathcal{S}_{:,i}) U$, for all $i$, via the operations depicted in Figure~\ref{fig:clifford_operations}.
These operations are analogous to the operations performed in the tableau representation~\cite{aaronson2004improved}.
We will say that $\tilde{\mathcal{S}} = U^\dag \mathcal{S} U$ if and only if $P(\tilde{\mathcal{S}}_{:,i}) = \pm U^\dag P(\mathcal{S}_{:,i}) U$ for all $i$ and for some Clifford operator $U$.

\begin{figure}[t]
    \centering
    \begin{subfigure}{0.3\textwidth}
        \centering
        \includegraphics{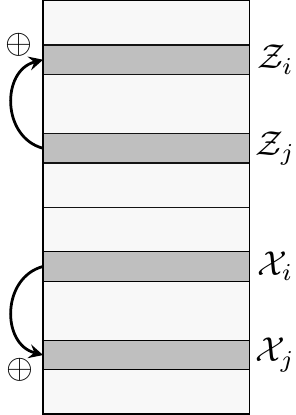}
        \caption{CNOT$_{i, j}$}
        \label{subfig:cnot}
    \end{subfigure}
    \begin{subfigure}{0.3\textwidth}
        \centering
        \includegraphics{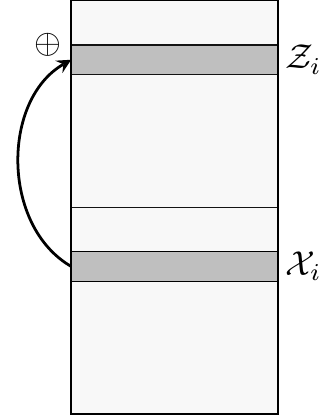}
        \caption{$S_i$}
        \label{subfig:s}
    \end{subfigure}
    \begin{subfigure}{0.3\textwidth}
        \centering
        \includegraphics{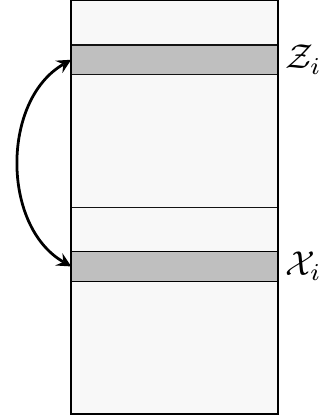}
        \caption{$H_i$}
        \label{subfig:h}
    \end{subfigure}
    \caption{
    Operations on a sequence of Pauli products $\mathcal{S}=\begin{bmatrix}\mathcal{Z}\\\mathcal{X}\end{bmatrix}$ corresponding to the conjugation of all its Pauli products by a Clifford gate.
For the CNOT$_{i,j}$ gate where $i$ is the control qubit and $j$ is the target qubit~\subref{subfig:cnot}, the $\mathcal{Z}_j$ and $\mathcal{X}_i$ rows are added to the $\mathcal{Z}_i$ and $\mathcal{X}_j$ rows respectively.
For a $S$ gate applied on qubit $i$~\subref{subfig:s}, the $\mathcal{X}_i$ row is added to the $\mathcal{Z}_i$ row.
For a $H$ gate applied on qubit $i$~\subref{subfig:h}, the $\mathcal{Z}_i$ and $\mathcal{X}_i$ rows are swapped.
}
    \label{fig:clifford_operations}
\end{figure}

Any Clifford$+R_Z$ circuit can be represented up to a global phase by a sequence of Pauli rotations followed by a final Clifford operator~\cite{gosset2014algorithm}.
The synthesis of a Pauli rotation is then a key procedure for constructing an equivalent Clifford$+R_Z$ circuit from this representation.
Let $U$ be a Clifford operator such that 
\begin{equation}\label{eq:pauli_conjugation}
U^\dag R_P(\theta)U = R_{U^\dag PU}(\theta) = R_{Z_i}(\theta)
\end{equation}
for some qubit $i$ and some Pauli operator $P \in \mathcal{P}^*_n$.
Then the synthesis of the Pauli rotation $R_P(\theta)$ can be performed by implementing $U$, $U^\dag$ and inserting a $R_Z(\theta)$ gate in between on qubit $i$.
If $P$ is diagonal then the Clifford operator $U$ satisfying Equation~\ref{eq:pauli_conjugation} can be implemented using only CNOT and $X$ gates.
Otherwise, if $P$ is not diagonal, at least one Hadamard gate is required to implement $U$ over the $\{X, \text{CNOT}, S, H\}$ gate set such that
\begin{equation}
U^\dag R_P(\theta)U = R_{P'}(\theta)
\end{equation}
where $P'$ is diagonal.
Note that the gate set considered is not minimal as the $X$ gate can be generated from the $S$ and $H$ gates.
As our cost model is the number of Hadamard gates we include the $X$ gate so that no $H$ gates are required to implement it.
The $X$ gate finds its purpose in the case where
\begin{equation}
U^\dag R_P(\theta)U = R_{U^\dag PU}(\theta) = R_{Z_i}(-\theta)
\end{equation}
by allowing the implementation of the $R_{Z}(-\theta)$ gate using the $R_{Z}(\theta)$ gate via the equality
\begin{equation}
R_{Z}(-\theta) = X R_{Z}(\theta) X.
\end{equation}
Nonetheless, in the case where $\theta = \pi/4$, the minimal $\{\mathrm{CNOT}, S, H\}$ gate set can be used since the negative sign can be compensated by inserting three $S$ gates as
\begin{equation}
UR_{Z_i}(7\pi/4)U^\dag = UR_{Z_i}(-\pi/4)U^\dag = R_{UZ_iU^\dag}(-\pi/4)=R_P(\pi/4).
\end{equation}

\subsection{Diagonalization network}

The synthesis of a sequence of Pauli rotations using the Clifford$+R_Z$ gate set implies the construction of a diagonalization network, derived from the notion of parity network established in Reference~\cite{amy2018controlled} and which is defined as follows.

\begin{definition}[Diagonalization network]
    A Clifford circuit $C$ is a diagonalization network for a sequence $\mathcal{S}$ of $m$ Pauli products if and only if there exists $m$ non-negative integers $\alpha_{0} \leq \ldots \leq \alpha_{m-1}$ such that $U_{i}^\dag P(\mathcal{S}_{:,i}) U_{i}$ is diagonal, where $U_{i}$ is the Clifford operator implemented by the first $\alpha_i$ gates of $C$.
\end{definition}

A sequence of $m$ Pauli rotations can be represented by a triple $(\mathcal{S}, \bs b, \bs \theta)$, where $\mathcal{S}$ encodes a sequence of $m$ Pauli products, $\bs b \in \{-1, 1\}^m$ and $\bs \theta \in \mathbb{R}^m$ such that $b_i$ and $\theta_i$ correspond to the sign and angle associated with the Pauli product $P(\mathcal{S}_{:,i})$.
Let $C$ be a diagonalization network for $\mathcal{S}$, then the sequence of Pauli rotations represented by $(\mathcal{S}, \bs b, \bs \theta)$ can be easily implemented from $C$ up to a final Clifford circuit by inserting $m$ $\{X, \text{CNOT}, R_Z\}$ subcircuits into $C$.
Indeed, as stated previously, if a Pauli product $P$ is diagonal then the Clifford operator $V$ satisfying 
\begin{equation}
V^\dag R_P(\theta)V = R_{V^\dag PV}(\theta) = R_{Z_j}(\theta)
\end{equation}
for some qubit $j$, can be implemented using only $\mathrm{CNOT}$ and $X$ gates.
And because $C$ is a diagonalization network for $\mathcal{S}$ then by definition there exists $m$ non-negative integers $\alpha_0 \leq \ldots \leq \alpha_{m-1}$ such that $U_i^\dag P(\mathcal{S}_{:,i})U_i$ is diagonal, where $U_i$ is the Clifford operator implemented by the first $\alpha_i$ gates of $C$.
It follows that inserting, for all $i$ and just after the $\alpha_i$th gate of $C$, a $\{\mathrm{CNOT}, X\}$ implementation of the Clifford operators $V_i$ and $V_i^\dag$ with the $R_{Z_j}(b_i \theta_i)$ gate in between, such that $V_i$ satisfies 
\begin{equation}
V_i^\dag U_i^\dag P(\mathcal{S}_{:,i})U_i V_i = R_{Z_j}(b_i \theta_i)
\end{equation}
for some qubit $j$, will result in an implementation of the sequence of Pauli rotations defined by $(\mathcal{S}, \bs b, \bs \theta)$ up to a final Clifford circuit.

The circuit obtained by this procedure obviously contains the same number of Hadamard gates as $C$ as no additional Hadamard gate was inserted.
Thus, synthesizing a sequence of Pauli rotations represented by $(\mathcal{S}, \bs b, \bs \theta)$ with a minimal number of Hadamard gates up to a final Clifford operator is equivalent to the problem of constructing a diagonalization network for $\mathcal{S}$ using a minimal number of Hadamard gates.
This approach can easily be extended to take into account the final Clifford operator, as explained in Section~\ref{sec:extension}.
We define $h(C)$ as being the number of Hadamard gates in a Clifford circuit $C$, and we extend the notation for a sequence of Pauli products $\mathcal{S}$ such that $h(\mathcal{S}) = \min\{h(C) \mid \text{$C$ is a diagonalization network for $\mathcal{S}$} \}$.
The problem of synthesizing a sequence of Pauli rotations ignoring the final Clifford operator with a minimal number of Hadamard gates can then be defined as follows.

\begin{problem}[H-Opt]
    Given a sequence $\mathcal{S}$ of Pauli products, find a Clifford circuit $C$ that is a diagonalization network for $\mathcal{S}$ and such that $h(C) = h(\mathcal{S})$.
\end{problem}

In Clifford$+T$ circuits, the Hadamard gadgetization procedure aims to transform the circuit in order to obtain an Hadamard-free subcircuit containing all the $T$ gates.
Hence, a Hadamard gate does not need to be gadgetized if there is no $T$ gate preceding it.
To take this particularity into consideration we define the following problem relating to the synthesis of a sequence of Pauli rotations up to a final Clifford circuit with a minimal number of internal Hadamard gates.

\begin{problem}[Internal-H-Opt]
    Given a sequence $\mathcal{S}$ of Pauli products, find a Clifford circuit $C = C_1 :: C_2$, i.e.\ $C$ is the circuit resulting from the concatenation of $C_1$ and $C_2$, such that $h(C_2)$ is minimized and $C_2$ is a diagonalization network for $\tilde{\mathcal{S}} = U^\dag \mathcal{S} U$ where $U$ is the Clifford operator associated with $C_1$.
\end{problem}

In Section~\ref{sec:diag_alg} we propose a diagonalization network synthesis algorithm to solve the H-Opt problem.
We prove its optimality in Section~\ref{sec:diag_opt}, and it is then employed in Section~\ref{sec:internal} to design an algorithm solving the Internal-H-Opt problem.

\section{Hadamard gates minimization}\label{sec:h_opt}

\subsection{Diagonalization network synthesis algorithm}\label{sec:diag_alg}

We first describe a simple procedure, of fundamental importance in our diagonalization network synthesis algorithm, to construct a Clifford operator $U$ such that $U^\dag P U$ is diagonal, where $P$ is a non-diagonal Pauli product.
Let $i$ be such that $P_i \in \{X, Y\}$, which necessarily exists as $P$ is non-diagonal.
If there exists $j \neq i$ such that $P_j \in \{X, Y\}$, then, based on the operation depicted in Figure~\ref{subfig:cnot}, we can deduce that the Pauli product $P'$ resulting from the conjugation of $P$ by the $\mathrm{CNOT}_{i, j}$ gate satisfies $P'_i \in \{X, Y\}$, $P'_j \in \{I, Z\}$ and $P'_k = P_k$ for all $k \not \in \{i, j\}$.
More generally, if $P' = U^\dag P U$ where $U$ is the Clifford operator associated with the fan-out formed by the gates $\{\mathrm{CNOT}_{i, j} \mid P_j \in \{X, Y\}, \forall j \neq i\}$, then $P'_j$ is diagonal for all $j \neq i$.
To complete the diagonalization of $P'$ we then just have to make $P'_i$ diagonal while preserving this property.
If $P'_i = Y$ then conjugating $P'$ by a $S$ gate on qubit $i$ maps $P'_i$ to $X$.
And in the case where $P'_i = X$, then conjugating $P'$ by a $H$ gate on qubit $i$ maps $P'_i$ to $Z$, and our diagonalization procedure is complete as the $S_i$ and $H_i$ operations do not affect $P'_j$ where $j \neq i$.

Consider the diagonalization network synthesis algorithm whose pseudo-code is given in Algorithm~\ref{alg:diagonalization} and which takes a sequence $\mathcal{S}$ of $m$ Pauli products as input.
The algorithm constructs a Clifford circuit $C$ iteratively by processing the Pauli products constituting $\mathcal{S}$ in order.
When a Pauli product $P = P(\mathcal{S}_{:,i})$ is being processed, if $U^\dag P U$, where $U \in \mathcal{C}_n$ is the Clifford operator implemented by $C$, is diagonal then the algorithm moves on to the next Pauli product.
Otherwise, if $U^\dag P U$ is not diagonal, a sequence of gates, constructed using the procedure described above, are appended to $C$ so that the updated Pauli product $U^\dag P U$ is diagonal.
Thus, Algorithm~\ref{alg:diagonalization} outputs a Clifford circuit that is a diagonalization network for $\mathcal{S}$.
A detailed execution example of Algorithm~\ref{alg:diagonalization} is provided in Appendix~\ref{app:example}.\\

\begin{algorithm}[t]
    \caption{Diagonalization network synthesis with a minimal number of $H$ gates}
    \label{alg:diagonalization}
	\SetAlgoLined
	\SetArgSty{textnormal}
	\SetKwFunction{proc}{DiagonalizationNetworkSynthesis}
	\SetKwInput{KwInput}{Input}
	\SetKwInput{KwOutput}{Output}
    \KwInput{A sequence $\mathcal{S} = \begin{bmatrix}\mathcal{Z}\\ \mathcal{X}\end{bmatrix}$ of $m$ Pauli products.}
    \KwOutput{A diagonalization network for $\mathcal{S}$ with a minimal number of $H$ gates.}
	\SetKwProg{Fn}{procedure}{}{}
    \Fn{\proc{$\mathcal{S}$}}{
		$C \leftarrow$ new empty circuit\\
        \If{$\mathcal{S}$ is empty}{
            \Return C
        }
        \If{$\exists i$ such that $\mathcal{X}_{i, 0} = 1$}{
            \ForEach{$j \in \{j \mid \mathcal{X}_{j,0} = 1\}\setminus \{i\}$}{
                $C \leftarrow C :: \mathrm{CNOT}_{i, j}$\\
                $\mathcal{S} \leftarrow \mathrm{CNOT}_{i, j}\,\mathcal{S}\,\mathrm{CNOT}_{i, j}$\\
            }
            \If{$\mathcal{Z}_{i,0} = 1$}{
                $C \leftarrow C :: S_i$\\
                $\mathcal{S} \leftarrow S_i\,\mathcal{S}\,S_i$\\
            }
            $C \leftarrow C :: H_i$\\
            $\mathcal{S} \leftarrow H_i\,\mathcal{S}\,H_i$\\
        }
        $\mathcal{S} \leftarrow \mathcal{S}$ with its first column removed\\
        \Return $C :: \texttt{DiagonalizationNetworkSynthesis}(\mathcal{S})$
	}
\end{algorithm}

\noindent{\bf Complexity analysis.}
At each iteration the algorithm carries out at most $\mathcal{O}(n)$ row operations on $\mathcal{S}$ where $n$ is the number of qubits, $m$ iterations are performed and $\mathcal{S}$ has $m$ columns, therefore the complexity of Algorithm~\ref{alg:diagonalization} is $\mathcal{O}(nm^2)$.

In the typical case where $n < m$, a faster version of Algorithm~\ref{alg:diagonalization} can be implemented using the tableau representation~\cite{aaronson2004improved}.
Let $\mathcal{T}$ be a tableau initialized at the begining of the algorithm.
Instead of updating $\mathcal{S}$ for each Clifford gate appended to the circuit $C$, we can use $\mathcal{T}$ to keep track of the Clifford operator $U$ implemented by $C$.
For each Clifford gate appended to $C$, $\mathcal{T}$ can be updated in $\mathcal{O}(n)$~\cite{aaronson2004improved}.
Then, the algorithm proceeds in the same way as Algorithm~\ref{alg:diagonalization} by sequentially diagonalizing the Pauli products represented by $\mathcal{S}$.
However, the $i$th Pauli product to diagonalized is not $P(\mathcal{S}_{:,i})$ but $U^\dag P(\mathcal{S}_{:,i}) U$, which can be computed in $\mathcal{O}(n^2)$ using the tableau $\mathcal{T}$.
This operation must be performed $\mathcal{O}(m)$ times and $\mathcal{T}$ must be updated $\mathcal{O}(nm)$ times as the number of gates in the final Clifford circuit is $\mathcal{O}(nm)$, therefore the overall time complexity of this algorithm is $\mathcal{O}(n^2m)$.
More details on this approach are given in Section~\ref{sec:improving}, where this algorithm is adapated to take a Clifford$+R_Z$ circuit as input instead of a sequence of Pauli products.\\

\noindent{\bf Hadamard gate count.}
In order to evaluate $h(C)$, where $C$ is the output circuit of Algorithm~\ref{alg:diagonalization}, we will rely on the following definition.

\begin{definition}[Commutativity matrix]
    Let $\mathcal{S}$ be a sequence of $m$ Pauli products. 
    The commutativity matrix $A^{(\mathcal{S})}$ associated with $\mathcal{S}$ is a strictly upper triangular Boolean matrix of size $m \times m$ such that for all $i < j$:
    $$
    A_{i, j}^{(\mathcal{S})}= 
    \begin{cases}
        0 &\text{ if $\mathcal{S}_{:,i}$ commutes with $\mathcal{S}_{:,j}$},\\ 
        1 &\text{ if $\mathcal{S}_{:,i}$ anticommutes with $\mathcal{S}_{:,j}$}.
    \end{cases}
    $$
\end{definition}

For convenience we will drop the superscript $(\mathcal{S})$ from $A$ when it is clear from the context that $A$ is associated with $\mathcal{S}$.
The commutativity matrix $A^{(\mathcal{S})}$ can also be seen as the adjacency matrix of a directed acyclic graph, which has already been studied and linked to the $T$-depth optimization problem~\cite{zhang2019optimizing}.
In this work, we further reinforce the interest in this structure by establishing a relation between the H-Opt and Internal-H-Opt problems and the rank of $A^{(\mathcal{S})}$.
Note that if $\tilde{\mathcal{S}} = U^\dag \mathcal{S} U$, where $U$ is some Clifford operator, then $A^{(\mathcal{\tilde{S}})} = A^{(\mathcal{S})}$ because if two Pauli products $P$ and $P'$ are commuting (or anticommuting) then $U^\dag P U$ and $U^\dag P' U$ are commuting (or anticommuting).

The number of Hadamard gates in the circuit produced by Algorithm~\ref{alg:diagonalization} can be characterized via the following theorem.

\begin{theorem}\label{thm:alg_h}
    Let $\mathcal{S} = \begin{bmatrix}\mathcal{Z} \\ \mathcal{X}\end{bmatrix}$ be a sequence of $m$ Pauli products, $A$ be its commutativity matrix and $C$ be the Clifford circuit returned by Algorithm~\ref{alg:diagonalization} when $\mathcal{S}$ is given as input.
    Then $h(C) = \rank(M)$ where $M = \begin{bmatrix} \mathcal{X} \\ A \end{bmatrix}$.
\end{theorem}
\begin{proof}
Let $\mathcal{S}^{(i)} = \begin{bmatrix}\mathcal{Z}^{(i)} \\ \mathcal{X}^{(i)}\end{bmatrix}$ be the sequence of Pauli products given as input to the $i$th recursive call of Algorithm~\ref{alg:diagonalization} with $\mathcal{S}^{(0)} = \mathcal{S}$, and let $M^{(i)} = \begin{bmatrix} \mathcal{X}^{(i)} \\ A^{(i)} \end{bmatrix}$ where $A^{(i)}$ is the commutativity matrix associated with $\mathcal{S}^{(i)}$.
We first start by analyzing how $M^{(i)}$ evolves to $M^{(i+1)}$ when $\mathcal{S}_{:,0}^{(i)}$ is diagonal.
In such case, we can obtain $M^{(i+1)}$ from $M^{(i)}$ by removing the first column of $M^{(i)}$ and the first row of its submatrix $A^{(i)}$.
Let $P = P(\mathcal{S}_{:,0}^{(i)})$, then, because $P$ is diagonal, the following equation holds:
\begin{equation}\label{eq:diag_sum}
    \bigoplus_{k \in K} \mathcal{X}_{k}^{(i)} = \bigoplus_{k \in K} M_{k}^{(i)} = A_0^{(i)}
\end{equation}
where $K = \{k \mid \mathcal{Z}_{k,0}^{(i)} = 1\}$.
Indeed, as $P$ is diagonal we necessarily have $P_k = Z$ for some $k$, and in the case where $P_j = I$ for all $j \neq k$ we have $\mathcal{X}_{k,j}^{(i)} = 1$ if and only if $\mathcal{S}_{:,0}^{(i)}$ anticommutes with $\mathcal{S}_{:,j}^{(i)}$, and so $\mathcal{X}_k^{(i)} = A^{(i)}_0$.
In a more general case, if there exists $j \neq k$ satisfying $P_j = Z$ then we can apply a $\mathrm{CNOT}_{j, k}$ gate for all such $j$ in order to fall back on our previous case, which implies Equation~\ref{eq:diag_sum}.
Consequently, removing the first row of the submatrix $A^{(i)}$ will not change the rank of $M^{(i)}$.
Moreover, due to the fact that $\mathcal{S}_{:,0}^{(i)}$ is diagonal, the first column of $M^{(i)}$ is equal to the null vector.
Therefore we have $\rank(M^{(i + 1)}) = \rank(M^{(i)})$ when $\mathcal{S}_{:,0}^{(i)}$ is diagonal.

In the case where $\mathcal{S}_{:,0}^{(i)}$ is not diagonal, Algorithm~\ref{alg:diagonalization} will apply a sequence of CNOT and $S$ gates followed by a single $H$ gate.
Let $\tilde{\mathcal{S}}^{(i)} = \begin{bmatrix}\tilde{\mathcal{Z}}^{(i)} \\ \tilde{\mathcal{X}}^{(i)}\end{bmatrix}$ be the sequence of Pauli products obtained by conjugating all Pauli products of $\mathcal{S}^{(i)}$ by this sequence of CNOT and $S$ gates, and let $\tilde{M}^{(i)} = \begin{bmatrix} \tilde{\mathcal{X}}^{(i)} \\ A^{(i)} \end{bmatrix}$.
Note that we have $\rank(\tilde{M}^{(i)}) = \rank(M^{(i)})$ as applying a $S$ or CNOT operation on $\mathcal{S}^{(i)}$ does not change the rank of $\mathcal{X}^{(i)}$.
Let $j$ be the qubit on which the Hadamard gate is applied, we must have $\tilde{M}_{j,0}^{(i)} = 1$ and $\tilde{M}_{k,0}^{(i)} = 0$ for all $k \neq j$, which implies that $\tilde{M}_j^{(i)}$ is independent from all the other rows of $\tilde{M}^{(i)}$.
Let $\hat{M}^{(i)} = \begin{bmatrix} \hat{\mathcal{X}}^{(i)} \\ A^{(i)} \end{bmatrix}$ where $\hat{\mathcal{S}}^{(i)} = \begin{bmatrix}\hat{\mathcal{Z}}^{(i)} \\ \hat{\mathcal{X}}^{(i)}\end{bmatrix}$ is obtained by conjugating all Pauli products of $\tilde{\mathcal{S}}^{(i)}$ by a Hadamard gate on qubit $j$, and notice that $\hat{M}_k^{(i)} = \tilde{M}_k^{(i)}$ for all $k \neq j$.
Analogously to Equation~\ref{eq:diag_sum}, since $\hat{\mathcal{S}}^{(i)}_{:,0}$ is diagonal, the following equation holds:
\begin{equation}
\bigoplus_{k \in \hat{K}} \hat{\mathcal{X}}_k^{(i)} = \bigoplus_{k \in \hat{K}} \hat{M}_{k}^{(i)} = A_0^{(i)}
\end{equation}
where $\hat{K} = \{k \mid \hat{\mathcal{Z}}_{k,0}^{(i)} = 1\}$.
Furthermore, as $j \in \hat{K}$, $\hat{M}_j^{(i)}$ can be expressed as follows:
\begin{equation}
\hat{M}_{j}^{(i)} = \bigoplus_{k \in \hat{K} \setminus \{j\}} \hat{M}_{k}^{(i)} \oplus A_0^{(i)} = \bigoplus_{k \in \tilde{K}} \tilde{M}_{k}^{(i)} \oplus A_0^{(i)}
\end{equation}
where $\tilde{K} = \{k \mid \tilde{\mathcal{Z}}_{k,0}^{(i)} = 1\} = \hat{K} \setminus \{j\}$.
It follows that $\hat{M}_j^{(i)}$ is a linear combination of the rows of $\tilde{M}^{(i)}$ whereas $\tilde{M}_j^{(i)}$ is an independent row, and so $\rank(\hat{M}^{(i)}) = \rank(\tilde{M}^{(i)}) - 1$.
After the Hadamard gate has been applied we end up in the same case as when $\mathcal{S}_{:,0}^{(i)}$ is diagonal, therefore we have $\rank(M^{(i+1)}) = \rank(\hat{M}^{(i)}) = \rank(M^{(i)}) - 1$.

We demonstrated that $\rank(M^{(i+1)}) = \rank(M^{(i)})$ when no Hadamard gate is applied at the $i$th recursive call, and that $\rank(M^{(i+1)}) = \rank(M^{(i)}) - 1$ if one Hadamard gate is applied.
Thus, the number of Hadamard gates in the Clifford circuit $C$ is equal to $\rank(M) - \rank(M^{(m)})$ where $m$ is the number of Pauli products in $\mathcal{S}$.
The sequence of Pauli products $\mathcal{S}^{(m)}$ is empty, hence $\rank(M^{(m)}) = 0$ and $h(C) = \rank(M)$.\\
\end{proof}

\subsection{Optimality}\label{sec:diag_opt}
In this section we demonstrate the optimality of Algorithm~\ref{alg:diagonalization} by proving the following theorem.

\begin{theorem}\label{thm:opt_h}
    Let $\mathcal{S} = \begin{bmatrix}\mathcal{Z} \\ \mathcal{X}\end{bmatrix}$ be a sequence of $m$ Pauli products, $A$ be its commutativity matrix and $C$ be a Clifford circuit that optimally solves the H-Opt problem for $\mathcal{S}$.
    Then $h{(C)} = \rank(M)$ where $M = \begin{bmatrix} \mathcal{X} \\ A \end{bmatrix}$.
\end{theorem}

Our proof of Theorem~\ref{thm:opt_h} rests on the following proposition, which puts an upper bound on the number of Hadamard gates required to simultaneously diagonalize a set of mutually commuting Pauli products.

\begin{proposition}\label{prop:co-diag}
    Let $\mathcal{S} = \begin{bmatrix}\mathcal{Z} \\ \mathcal{X}\end{bmatrix}$ be a sequence of $m$ mutually commuting Pauli products of size $n$ and let $U \in \mathcal{C}_n$ be a Clifford operator such that $U^\dag P(\mathcal{S}_{:,i}) U$ is diagonal for all $i$.
    Then $h(C) \geq \rank(\mathcal{X})$, where $C$ is a Clifford circuit implementing $U$.
\end{proposition}

\begin{proof}
Let $\mathcal{S}^{(i)}$ be the state of $\mathcal{S}$ resulting from conjugating all its Pauli products by the Clifford operator implemented by the first $i$ gates of $C$.
If the $(i + 1)$th gate of $C$ is a CNOT or $S$ gate, then $\rank(\mathcal{X}^{(i + 1)}) = \rank(\mathcal{X}^{(i)})$.
Else, if the $(i + 1)$th gate of $C$ is a Hadamard gate, then $\mathcal{X}^{(i+1)}$ and $\mathcal{X}^{(i)}$ have at least $n-1$ rows in common and $1 \geq |\rank(\mathcal{X}^{(i)}) - \rank(\mathcal{X}^{(i+1)})| \geq 0$.
Therefore, the number of Hadamard gates in $C$ is at least $|\rank(\mathcal{X}) - \rank(\mathcal{X}^{(k)})|$, where $k$ is the number of gates in $C$.
The circuit $C$ performs a simultaneous diagonalization of the Pauli products constituting $\mathcal{S}$, which implies that $\rank(\mathcal{X}^{(k)}) = 0$, hence $h(C) \geq |\rank(\mathcal{X}) - \rank(\mathcal{X}^{(k)})| = \rank(\mathcal{X})$.\\
\end{proof}

In the following we use $\mathcal{S}_{:,:j}$ to denote the submatrix formed by the first $j$ columns of $\mathcal{S}$.
Theorem~\ref{thm:alg_h} implies an upper bound on the number of Hadamard gates required to solve the H-Opt problem.
There always exists a Clifford circuit $C$ that is a diagonalization network for a sequence of Pauli products $\mathcal{S} = \begin{bmatrix}\mathcal{Z} \\ \mathcal{X}\end{bmatrix}$ such that $\rank(M) \geq h(C)$ where $M = \begin{bmatrix} \mathcal{X} \\ A \end{bmatrix}$ and $A$ is the commutativity matrix associated with $\mathcal{S}$.
In order to prove Theorem~\ref{thm:opt_h} it remains to show that if $C$ is a diagonalization network for $\mathcal{S}$ then $h(C) \geq \rank(M)$.
To do so, we will show that we can derive a Clifford circuit $C'$ from $C$ such that $C'$ is satisfying $h(C') = h(C)$ and is a solution to a specific instance $\mathcal{S}'$ of the simultaneous diagonalization problem, where $\mathcal{S}' = \begin{bmatrix} \bs 0 \\ M' \end{bmatrix}$ is a sequence of mutually commuting Pauli products satisfying $\rank(M') \geq \rank(M)$.
By Proposition~\ref{prop:co-diag} we then would have $h(C) = h(C') \geq \rank(M') \geq \rank(M)$.
We first give a construction for $M'$ and prove that $\rank(M') \geq \rank(M)$ via the following proposition.

\begin{proposition}\label{prop:rank}
    Let $C$ be a diagonalization network for a sequence $\mathcal{S} = \begin{bmatrix}\mathcal{Z} \\ \mathcal{X}\end{bmatrix}$ of $m$ Pauli products of size $n$, and let $C^{(i)}$ be the subcircuit of $C$ truncated after its $i$th Hadamard gate.
    And let the matrices $\mathcal{S}^{(i)} = \begin{bmatrix}\mathcal{Z}^{(i)} \\ \mathcal{X}^{(i)}\end{bmatrix}$ be such that $\mathcal{S}^{(0)} = \mathcal{S}$ and in the case where $i>0$:
    \begin{align*} 
        \mathcal{S}^{(i)}_{:,j} &= \bs 0 && \text{if $C^{(i-1)}$ is a diagonalization network for $\mathcal{S}_{:,:j}$},\\
        P(\mathcal{S}^{(i)}_{:,j}) &= \pm U_{(i)}^\dag P(\mathcal{S}_{:,j}) U_{(i)} && \text{otherwise},
    \end{align*}
    where $U_{(i)} \in \mathcal{C}_n$ is the Clifford operator associated with $C^{(i)}$.
    Consider the matrices $M = \begin{bmatrix} \mathcal{X} \\ A \end{bmatrix}$ and $M' = \begin{bmatrix} \mathcal{X} \\ A' \end{bmatrix}$ where $A$ is the commutativity matrix associated with $\mathcal{S}$, and $A'$ is a matrix composed of $h(C)$ rows such that $A'_{i-1} = \mathcal{X}^{(i)}_j$ where $j$ is the qubit on which the $i$th Hadamard gate of $C$ is applied.
    Then we have $\rank(M') \geq \rank(M)$.
\end{proposition}

\begin{proof}
We will prove Proposition~\ref{prop:rank} by showing that the rows of $M$ are in the span of the set formed by the rows of the $\mathcal{X}^{(i)}$ matrices, which are themselves in the row space of the matrix $M'$.
We first show that the rows of $\mathcal{X}^{(i)}$ are in the span of the set formed by the first $n + i$ rows of $M'$.
For $i = 0$, this assertion is obviously true as we have $\mathcal{X}_j = M'_j$ for all $0 \le j < n$.
Let $\tilde{\mathcal{S}}^{(i)} = \begin{bmatrix}\tilde{\mathcal{Z}}^{(i)} \\ \tilde{\mathcal{X}}^{(i)}\end{bmatrix} = U^\dag \mathcal{S}^{(i)} U$, where $U$ is the Clifford operator implemented by the $\{\mathrm{CNOT}, S\}$ subcircuit comprised between the $i$th and $(i+1)$th Hadamard gate of $C$.
Note that performing a $\mathrm{CNOT}$ or $S$ operation on $\mathcal{S}^{(i)}$ doesn't change the row space of $\mathcal{X}^{(i)}$, therefore the rows of $\tilde{\mathcal{X}}^{(i)}$ are in the row space of $\mathcal{X}^{(i)}$.
Let $\bs \alpha$ be a vector of size $m$ such that $\alpha_j$ is the smallest non-negative integer for which $C^{(\alpha_j)}$ is a diagonalization network for $\mathcal{S}_{:,:j}$ and let $k$ be the qubit on which the $(i+1)$th Hadamard gate of $C$ is applied.
We can then distinguish 3 cases for the value of $\alpha_j$, where $0 \leq j < m$:
\begin{itemize}
    \item $\alpha_j < i $, then $\mathcal{S}^{(i+1)}_{:,j} = \tilde{\mathcal{S}}^{(i)}_{:,j} = \bs 0$ for all $j$, because if $C^{(i-1)}$ is a diagonalization network for $\mathcal{S}_{:,:j}$ then $C^{(i)}$ is also a diagonalization network for $\mathcal{S}_{:,:j}$;
    \item $\alpha_j = i$, then $\mathcal{S}^{(i + 1)}_{:,j} = \bs 0$ and $\tilde{\mathcal{S}}^{(i)}_{:,j}$ is diagonal, therefore $\mathcal{X}^{(i+1)}_{:,j} = \tilde{\mathcal{X}}^{(i)}_{:,j} = \bs 0$;
    \item $\alpha_j > i$ then $\tilde{\mathcal{S}}^{(i)}_{:,j}$ encodes the same Pauli product as $\mathcal{S}^{(i+1)}_{:,j}$ up to a Hadamard operation on qubit $k$, therefore $\mathcal{X}^{(i+1)}_{\ell,j} = \tilde{\mathcal{X}}^{(i)}_{\ell,j}$ for all $\ell \neq k$.
\end{itemize}
To sum up we have $\mathcal{X}^{(i + 1)}_j = \tilde{\mathcal{X}}^{(i)}_j$ for all $j \neq k$ and by definition we have $\mathcal{X}^{(i + 1)}_k = M'_{n+i+1}$.
It follows that the rows of $\mathcal{X}^{(i + 1)}$ are in the span of the set formed by the first $n + i + 1$ rows of $M'$.
We now show that, for all $j$, $A_j$ is in the row space of $\mathcal{X}^{(\alpha_j)}$.
Since $\mathcal{S}^{(\alpha_j)}_{:,j}$ is diagonal, similarly to Equation~\ref{eq:diag_sum}, the following holds:
\begin{equation}
    \bigoplus_{k \in K} \mathcal{X}^{(\alpha_j)}_{k,\ell}= 
    \begin{cases}
    0 & \text{if $\mathcal{S}_{:,\ell}$ commutes with $\mathcal{S}_{:,j}$ or } \ell \leq j,\\
    1 & \text{if $\mathcal{S}_{:,\ell}$ anticommutes with $\mathcal{S}_{:,j}$}, \forall \ell > j,
    \end{cases}
\end{equation}
where $K = \{k \mid \mathcal{Z}_{k, j}^{(\alpha_j)} = 1 \}$.
This sum satisfy the same properties as the row $j$ of the commutativity matrix $A$ associated with $\mathcal{S}$, therefore we have:
\begin{equation}
    A_j= \bigoplus_{k \in K} \mathcal{X}^{(\alpha_j)}_k.
\end{equation}
Thus, for all $j$, $A_j$ is in the row space of the matrix $\mathcal{X}^{(\alpha_j)}$, whose rows are themselves in the row space of $M'$.
Consequently, $M_j$ is in the row space of $M'$ for all $j$ and $\rank(M') \geq \rank(M)$.\\
\end{proof}

\begin{figure}[t]
    \begin{subfigure}[T]{0.5\textwidth}
    \begin{quantikz}[column sep=0.4cm, row sep=0.35cm]
        & \targ{} & \gate{S} & \gate{H} & \ctrl{1} & \gate{H} & \gate{S} & \targ{0} & \qw & \qw \\
        & \ctrl{-1} & \qw & \qw & \targ{} & \qw & \qw & \ctrl{-1} & \gate{H} & \qw \\
    \end{quantikz}\vspace{-0.3cm}
    \caption{Circuit $C$ that is a diagonalization network for $\mathcal{S}$ where $\mathcal{S}_{:,0}$ is diagonalized by the first $3$ gates, $\mathcal{S}_{:,1}$ by the first $5$ gates and $\mathcal{S}_{:,2}$ by the whole circuit.}
    \label{subfig:init_circ}
    \end{subfigure}
    \begin{subfigure}[T]{0.5\textwidth}
    \begin{quantikz}[column sep=0.4cm, row sep=0.35cm]
        & \targ{} & \gate{H} & \swap{2} & \ctrl{1} &\gate{H} & \swap{3} & \targ{0} & \qw & \qw & \qw \\
        & \ctrl{-1} & \qw & \qw & \targ{} & \qw & \qw & \ctrl{-1} & \gate{H} & \swap{3} & \qw \\
        & \qw & \qw & \targX{} & \qw & \qw & \qw & \qw & \qw & \qw & \qw \\
        & \qw & \qw & \qw & \qw & \qw & \targX{} & \qw & \qw & \qw & \qw \\
        & \qw & \qw & \qw & \qw & \qw & \qw & \qw & \qw & \targX{} & \qw\\
    \end{quantikz}\vspace{-0.3cm}
    \caption{Circuit $C'$ derived from $C$.\\~\\}
    \label{subfig:swap_circ}
    \end{subfigure}
    \begin{subfigure}[b]{0.5\textwidth}
    \begin{equation*}
        \arraycolsep=4.0pt
        \mathcal{S} = \begin{bmatrix}\mathcal{Z} \\ \mathcal{X}\end{bmatrix} =
        \left(\begin{array}{ccc}
                1 & 0 & 0 \\
                1 & 1 & 0 \\\hline
                1 & 0 & 1 \\
                0 & 1 & 0 \\
        \end{array}\right),\quad
        A^\mathcal{(S)} = 
        \left(\begin{array}{ccc}
                0 & 1 & 1 \\
                0 & 0 & 0 \\
                0 & 0 & 0 \\
        \end{array}\right)
    \end{equation*}
    \caption{\centering Initial sequence of Pauli products and its commutativity matrix.}
    \label{subfig:init_matrix}
    \end{subfigure}
    \begin{subfigure}[b]{0.5\textwidth}
    \begin{equation*}
        \arraycolsep=4.0pt
        \mathcal{S}' = \begin{bmatrix}\bs 0 \\ M'\end{bmatrix} = 
        \left(\begin{array}{ccc}
                \bs 0 & \bs 0 & \bs 0 \\\hline
                1 & 0 & 1 \\
                0 & 1 & 0 \\
                0 & 1 & 1 \\
                0 & 0 & 1 \\
                0 & 0 & 0 \\
        \end{array}\right)
    \end{equation*}
    \caption{\centering Pauli products that are simultaneously diagonalized by $C'$, where $M'$ is as defined in Proposition~\ref{prop:rank}.}
    \label{subfig:swap_matrix}
    \end{subfigure}
    \caption{Construction example of $C'$ and $\mathcal{S}'$ for the proof of Theorem~\ref{thm:opt_h}.}
    \label{fig:swap_example}
\end{figure}

The proof of Theorem~\ref{thm:opt_h} can now be formulated based on Proposition~\ref{prop:co-diag} and~\ref{prop:rank}.
\begin{proof}[Proof of Theorem~\ref{thm:opt_h}]
Consider a Clifford $C'$ containing the same number of Hadamard gates as $C$, acting over $n + h(C)$ qubits and constructed by the following process:
\begin{enumerate}
    \item Start with $C'$ as a copy of $C$ with $h(C)$ additional qubits.
    \item Remove all $S$ gates from $C'$.
    \item After the $i$th Hadamard gate of $C'$, insert a SWAP gate operating over the qubits $n+i-2$ and $j$ where $j$ is the qubit on which the $i$th Hadamard gate is applied.
\end{enumerate}
An example of this process is provided in Figure~\ref{fig:swap_example}.
The SWAP$_{i, j}$ gate can be implemented using $\mathrm{CNOT}$ gates:
$$\mathrm{SWAP}_{i,j} = \mathrm{CNOT}_{i,j}\mathrm{CNOT}_{j,i}\mathrm{CNOT}_{i,j}$$
This operation can be performed on $\mathcal{S}$ by swapping the rows $\mathcal{Z}_i$ and $\mathcal{Z}_j$ as well as the rows $\mathcal{X}_i$ and $\mathcal{X}_j$.
Let $M'$ be defined as in Proposition~\ref{prop:rank}, let $\mathcal{S}' = \begin{bmatrix} \bs 0 \\ M' \end{bmatrix}$ be a sequence of $m$ mutually commuting Pauli products of size $n + h(C)$ and let $U'$ is the Clifford operator implemented by $C'$, we will show that $U'^\dag P(\mathcal{S}'_{:,i})U'$ is diagonal for all $i$.
We reuse $C^{(i)}$ and $\mathcal{S}^{(i)} = \begin{bmatrix}\mathcal{Z}^{(i)} \\ \mathcal{X}^{(i)}\end{bmatrix}$ as defined in Proposition~\ref{prop:rank}, and we define $\mathcal{S}'^{(i)} = \begin{bmatrix}\mathcal{Z}'^{(i)} \\ \mathcal{X}'^{(i)}\end{bmatrix}$ and $C'^{(i)}$ analogously where $C'^{(i)}$ is the subcircuit resulting from truncating $C'$ after its $i$th inserted SWAP gate and $C'^{(0)}$ is the empty circuit.

We now prove by induction that for all $0 \leq i \leq h(C)$, $0 \leq j < n$ and $n \leq k < n + i$ we have $\mathcal{X}'^{(i)}_j = \mathcal{X}^{(i)}_j$, $\mathcal{Z}'^{(i)}_j = \bs 0$ and $\mathcal{X}'^{(i)}_k = \bs 0$.
For $i = 0$ and $0 \leq j < n$, the equalities $\mathcal{X}'_j = \mathcal{X}_j$ and $\mathcal{Z}'^{(i)}_j = \bs 0$ are satisfied by definition.
Let $0 \leq i < h(C)$ and $\alpha_i$ be the qubit on which the $(i+1)$th Hadamard gate of $C$ is applied.
The matrix $\mathcal{S}^{(i+1)}$ can be obtained from $\mathcal{S}^{(i)}$ by performing a sequence of $\{\mathrm{CNOT}, S\}$ operations and a Hadamard operation on qubit $\alpha_i$.
Similarly, the matrix $\mathcal{S}'^{(i+1)}$ can be obtained from $\mathcal{S}'^{(i)}$ by performing the same sequence of $\mathrm{CNOT}$ operations, a Hadamard operation on qubit $\alpha_i$ and a SWAP operation acting on the qubits $\alpha_i$ and $n+i$.
In both cases, the rows $\mathcal{X}^{(i)}_j$ and $\mathcal{X}'^{(i)}_j$, where $0 \leq j < n, j \neq \alpha_i$, are only affected by the $\mathrm{CNOT}$ operations, and so if $\mathcal{X}'^{(i)}_j = \mathcal{X}^{(i)}_j$ for all $0 \leq j < n$, then $\mathcal{X}'^{(i+1)}_j = \mathcal{X}^{(i+1)}_j$ for all $0 \leq j < n, j \neq \alpha_i$.
Notice that the only gate in $C'$ acting on the qubit $n + i$ is the SWAP gate operating on the qubits $\alpha_i$ and $n + i$ and recall that by definition $\mathcal{X}'_{n+i} = \mathcal{X}^{(i+1)}_{\alpha_i}$; then, because this SWAP gate is the last gate of the circuit $C'^{(i+1)}$, we have $\mathcal{X}'^{(i+1)}_{\alpha_i} = \mathcal{X}'_{n+i} = \mathcal{X}^{(i+1)}_{\alpha_i}$.
Therefore, for all $0 \leq j < n$, if $\mathcal{X}'^{(i)}_j = \mathcal{X}^{(i)}_j$ then $\mathcal{X}'^{(i+1)}_j = \mathcal{X}^{(i+1)}_j$.

If $\mathcal{Z}'^{(i)}_j = \bs 0$ for all $0 \leq j < n$, then applying a sequence of CNOT operations on $\mathcal{S}'^{(i)}$ acting on the first $n$ qubits will not alter the matrix $\mathcal{Z}'^{(i)}$.
Thus, if $\mathcal{Z}'^{(i)}_j = \bs 0$ for all $0 \leq j < n$, then $\mathcal{Z}'^{(i+1)}_j = \mathcal{Z}'^{(i)}_j = \bs 0$ for all $0 \leq j < n, j \neq \alpha_i$.
Furthermore, if $\mathcal{Z}'^{(i)}_{\alpha_i} = \bs 0$ for all $0 \leq j < n$, then applying a Hadamard operation on $\mathcal{S}'^{(i)}$ acting on qubit $\alpha_i$ after this sequence of CNOT operations would yield $\mathcal{X}'^{(i)}_{\alpha_i} = \bs 0$.
This Hadamard operation is followed by a SWAP operation between the qubits $\alpha_i$ and $k = n + i$ which would induce that $\mathcal{X}'^{(i + 1)}_k = \bs 0$ and $\mathcal{Z}'^{(i + 1)}_{\alpha_i} = \bs 0$ because $\mathcal{Z}'_{k} = \bs 0$.
Thus, for all $0 \leq j < n$ , if $\mathcal{Z}'^{(i)}_j = \bs 0$ then  $\mathcal{Z}'^{(i+1)}_j = \bs 0$.
In addition, for all $k$ such that $n \leq k < n + i$, the circuit $C'^{(i+1)}$ doesn't contain any gate operating on the qubit $k$ other than those included in $C'^{(i)}$; therefore if $\mathcal{X}'^{(i)}_k = \bs 0$ for all $n \leq k < n + i$ then $\mathcal{X}'^{(i + 1)}_k = \bs 0$ for all $n \leq k < n + i + 1$.

Let $i = h(C)$, by combining the facts that $\mathcal{X}'^{(i)}_j = \mathcal{X}^{(i)}_j = \bs 0$ for all $0 \leq j < n$ and $\mathcal{X}'^{(i)}_{n+j} = \bs 0$ for all $0 \leq j < i$, we can deduce that $\mathcal{X}'^{(i)}$ is the null matrix which imply that $U'^\dag P(\mathcal{S}'_{:,j})U'$ is diagonal for all $j$ where $U'$ is the Clifford operator implemented by $C'$.
By Proposition~\ref{prop:co-diag} we have $h(C) = h(C') \ge \rank(M')$, and by Proposition~\ref{prop:rank} we have $\rank(M') \ge \rank(M)$ which entails $h(C) \ge \rank(M)$.
This lower bound is satisfied by Algorithm~\ref{alg:diagonalization} as stated by Theorem~\ref{thm:alg_h}, this implies that Algorithm~\ref{alg:diagonalization} is optimal and concludes the proof of Theorem~\ref{thm:opt_h}.\\
\end{proof}

\noindent{\bf Pauli rotations ordering.}
Algorithm~\ref{alg:diagonalization} solves the H-Opt problem for a fixed sequence of Pauli rotations.
However, if two adjacent Pauli rotations are commuting then their order could be inverted, leading to another sequence of Pauli rotations representing the same operator.
We show that changing the ordering in this way doesn't affect the minimal number of Hadamard gate required to implement the diagonalization network associated with the sequence of Pauli rotations.
Let $\mathcal{S} = \begin{bmatrix}\mathcal{Z} \\ \mathcal{X}\end{bmatrix}$ be a sequence of Pauli products, let $i$ be such that $\mathcal{S}_{:,i}$ commutes with $\mathcal{S}_{:,i+1}$ and let $\mathcal{S}' = \begin{bmatrix}\mathcal{Z}' \\ \mathcal{X}'\end{bmatrix}$ be a sequence of Pauli products obtained by swapping the columns $i$ and $i+1$ of $\mathcal{S}$.
Let $M = \begin{bmatrix} \mathcal{X} \\ A \end{bmatrix}$  and $M' = \begin{bmatrix} \mathcal{X}' \\ A' \end{bmatrix}$ where $A$ and $A'$ are the commutativity matrices of $\mathcal{S}$ and $\mathcal{S}'$ respectively.
Since $\mathcal{S}_{:,i}$ commutes with $\mathcal{S}_{:,i+1}$ we have $A_{i, i+1} = A'_{i, i+1} = 0$, and so $M_{:, i} = M'_{:, i+1}$ and $M_{:, i+1} = M'_{:, i}$.
The matrix $M'$ can be obtained from $M$ by swapping its columns $i$ and $i+1$, which entails $\rank(M) = \rank(M')$.
Thus, inverting the order of two adjacent and commuting Pauli rotations doesn't change the minimal number of Hadamard gates required to implement the diagonalization network associated with the sequence of Pauli rotations.\\

\noindent{\bf Other gate sets.}
One could consider the problem over other Clifford gate sets, which raises the question of whether these gate sets could perform better than the $\{X, \text{CNOT}, S, H\}$ gate set considered.
In order to achieve a number of Hadamard gate inferior to $\rank(M)$, where $M$ is defined as in Theorem~\ref{thm:opt_h}, the gate set considered needs to have at least one gate, other than the Hadamard gate, such that its decomposition over the $\{X, \mathrm{CNOT}, S, H\}$ gate set necessarily involves at least one Hadamard gate.
Said otherwise, the number of Hadamard gates is at least $\rank(M)$ for any gate set in which the Hadamard gate is the only gate $U$ for which there exists a non-diagonal Pauli operator $P$ such that $U^\dag PU$ is diagonal.

\subsection{Extension to Clifford$+R_Z$ circuit re-synthesis} \label{sec:extension}
Any Clifford$+R_Z$ circuit can be characterized by a sequence of Pauli rotations followed by a final Clifford operator $C_f$~\cite{gosset2014algorithm}.
We demonstrated that Algorithm~\ref{alg:diagonalization} solves the H-Opt problem optimally, and so it can be used to synthesize a sequence of Pauli rotations up to a final Clifford operator $C_{f'}$ with a minimal number of Hadamard gates.
The synthesis of the full Clifford$+R_Z$ circuit can then be performed by coupling Algorithm~\ref{alg:diagonalization} with a procedure to synthesize the Clifford operator $C_f \cdot C_{f'}$.
We will demonstrate that this procedure can in fact also be performed by Algorithm~\ref{alg:diagonalization} with a minimal number of Hadamard gates.

A Clifford operator $U \in \mathcal{C}_n$ can be represented by a tableau encoding $2n$ Pauli operators such that $n$ of them are mutually commuting Pauli operators called stabilizer generators and the other half are also mutually commuting Pauli operators referred to as destabilizer generators.
If the stabilizer generators are all diagonalized, then the Clifford operator can be synthesized using only $\{X, S ,\mathrm{CNOT}\}$ gates~\cite{aaronson2004improved}.
Thus, synthesizing a Clifford operator with the minimal number of Hadamard gates amounts to finding a Clifford circuit $C$ containing the minimal number of Hadamard gates and such that $U^\dag P U$ is diagonal for all $P$ in the stabilizer generators, where $U$ is the Clifford operator associated with $C$.
We will demonstrate via the following proposition that a Clifford circuit satisfying these properties is produced by Algorithm~\ref{alg:diagonalization} when the sequence of Pauli products $\mathcal{S}$ given as input encodes the stabilizer generators on any order.

\begin{proposition}\label{prop:co-diag_opt}
    Let $\mathcal{S}$ be a sequence of $m$ mutually commuting Pauli products and $C$ be the Clifford circuit returned by Algorithm~\ref{alg:diagonalization} when $\mathcal{S}$ is given as input.
    Then $U^\dag P(\mathcal{S}_{:,j}) U$ is diagonal for all $j$, where $U$ is the Clifford operator associated with $C$.
\end{proposition}

\begin{proof}
Let $P$ and $P'$ be commuting Pauli operators such that $P$ is diagonal and $P'$ is not diagonal.
If there exists $k$ such that $P'_k = X$ and $P'_\ell \in \{I, Z\}$ for all $\ell \neq k$, then $P_k = I$ because $P$ commutes with $P'$ and $P$ is diagonal.
Therefore conjugating $P$ and $P'$ with a Hadamard gate on qubit $k$ will result in both operators being diagonalized.
Let $C^{(i)}$ be the subcircuit of $C$ truncated before its $i$th Hadamard gate with $C^{(0)}$ defined as the empty circuit, and let $U_{(i)}$ be the Clifford operator associated with $C^{(i)}$.
Due to the construction process of $C$, for each subcircuit $C^{(i)}$ where $i>0$ there exists $j$ such that $P' = U_{(i)}^\dag \mathcal{S}_{:,j} U_{(i)}$ satisfies $P'_k = X$ and $P'_\ell \in \{I, Z\}$ for all $\ell \neq k$ where $k$ is the qubit on which the $i$th Hadamard gate of $C$ is applied.
Hence, for all $i < h(C)$ and for all $j$, if $U_{(i)}^\dag P(\mathcal{S}_{:,j}) U_{(i)}$ is diagonal, then $U_{(i+1)}^\dag  P(\mathcal{S}_{:,j}) U_{(i+1)}$ is also diagonal.
The circuit $C$ is a diagonalization network for $\mathcal{S}$, which imply that for all $j$ there exists $U_{(i)}$ such that $U_{(i)}^\dag P(\mathcal{S}_{:,j}) U_{(i)}$ is diagonal, and so $U^\dag P(\mathcal{S}_{:,j}) U$ is a diagonal for all $j$.\\
\end{proof}

Based on Proposition~\ref{prop:co-diag_opt}, we can now show that Algorithm~\ref{alg:diagonalization} can be used to synthesize a sequence of Pauli rotations followed by a final Clifford operator with a minimal number of Hadamard gates.
Let $\mathcal{S}$ be a sequence of Pauli products associated with the sequence of Pauli rotations we are aiming to implement, let $\mathcal{S}'$ be a sequence of Pauli products encoding the stabilizer generators of the final Clifford operator, and let $\tilde{\mathcal{S}} = \begin{bmatrix} \mathcal{S} & \mathcal{S}' \end{bmatrix}$.
Any $\{X, \mathrm{CNOT}, S, H, R_Z\}$ circuit implementing this sequence of Pauli rotations followed by the final Clifford operator is necessarily a diagonalization network for $\tilde{\mathcal{S}}$.
The circuit $C$ returned by Algorithm~\ref{alg:diagonalization} when $\tilde{\mathcal{S}}$ is given as input satisfies this condition with a minimal number of Hadamard gates.
Moreover, as indicated by Proposition~\ref{prop:co-diag_opt}, $C$ simultaneously diagonalize the sequence of Pauli products encoded by $\mathcal{S}'$.
Thus, the synthesis of the sequence of Pauli rotations and the final Clifford operator can be completed with a minimal number of Hadamard gate by inserting $\{X, \mathrm{CNOT}, S, R_Z\}$ subcircuits into $C$.

Let $C$ be the circuit obtained once the number of Hadamard gates have been optimized with our method.
The circuit $C$ may contain an important number of CNOT gates as our algorithm does not aim at optimizing the CNOT-count.
If necessary, several methods can be used to optimize the number of CNOT gates in $C$ while preserving the number of Hadamard gates.
First, the Clifford parts of $C$ can be re-synthesized by using a Clifford circuit synthesis algorithm that preserve the optimal number of Hadamard gates.
For example, as shown in Reference~\cite{maslov2018shorter}, a Clifford circuit can be implemented with the optimal number of Hadamard gates via two $\{\mathrm{CNOT}, \mathrm{CZ}, S\}$ circuits separated by a layer of Hadamard gates.
Algorithms designed for the synthesis of $\{\mathrm{CNOT}, \mathrm{CZ}, S\}$ circuits can then be used to optimize the number of gates or the depth of the circuit~\cite{de2022graph}.
A complementary way of optimizing the number of CNOT gates in $C$ is to re-synthesize the Hadamard-free subcircuits of $C$ via a phase polynomial synthesis algorithm~\cite{amy2018controlled, vandaele2022phase}.
This method would probably be more effective than the re-synthesis of the Clifford parts approach when $C$ contains large Hadamard-free subcircuits.

\section{Internal Hadamard gates minimization}\label{sec:internal}
In this section, we tackle the problem of minimizing the number of internal Hadamard gates, which corresponds to the number of Hadamard gates occurring between the first and the last non-Clifford $R_Z$ gate of the circuit.
We first give an algorithm in Section~\ref{sec:algorithm_internal} that performs the synthesis of a diagonalization network while minimizing the number of internal Hadamard gates.
We then prove its optimality in Section~\ref{sec:h_internal_opt}.

\subsection{Algorithm}\label{sec:algorithm_internal}
Solving the Internal-H-Opt problem for a sequence $\mathcal{S} = \begin{bmatrix}\mathcal{Z} \\ \mathcal{X}\end{bmatrix}$ of Pauli products consists in finding a Clifford operator $U$ such that $\rank(\tilde{M})$ is minimal where $\tilde{M} = \begin{bmatrix} \tilde{\mathcal{X}} \\ A \end{bmatrix}$, $\tilde{\mathcal{S}} = \begin{bmatrix}\tilde{\mathcal{Z}} \\ \tilde{\mathcal{X}}\end{bmatrix} = U^\dag\mathcal{S}U$ and $A$ is the commutativity matrix associated with $\mathcal{S}$.
The inequality $\rank(A) \leq \rank(\tilde{M}) \leq \rank(M) \leq \rank(A) + n$, where $M = \begin{bmatrix} \mathcal{X} \\ A \end{bmatrix}$ and $n$ is the number of qubits, imply that the circuit produced by Algorithm~\ref{alg:diagonalization} contains at most $n$ additional internal Hadamard gates when compared to an optimal solution.
To go beyond this approximation and obtain an optimal solution, it is necessary to find a sequence of Clifford operations which, when applied to $\mathcal{S}$, transform $\mathcal{X}$ into $\tilde{\mathcal{X}}$.
As discussed in Section~\ref{sec:extension}, implementing a Clifford operator can be done in two parts: finding a circuit that simultaneously diagonalize the stabilizer generators of the Clifford operator and finishing the implementation with a $\{X, S, \mathrm{CNOT}\}$ circuit.
The $\{X, S, \mathrm{CNOT}\}$ circuit can be disregarded as the associated operations have no impact on the rank of $\tilde{M}$.
Hence, solving the Internal-H-Opt problem for a sequence $\mathcal{S}$ of Pauli products consists in finding a set of mutually commuting Pauli products, encoded in a matrix $\mathcal{S}'$, that are simultaneously diagonalized by a Clifford operator $U$ and such that $\rank(\tilde{M})$ is minimal where $\tilde{M} = \begin{bmatrix} \tilde{\mathcal{X}} \\ A \end{bmatrix}$ and $\tilde{\mathcal{S}}$ is the sequence of Pauli products resulting from conjugating all the Pauli products of $\mathcal{S}$ by $U$.
As stated by Proposition~\ref{prop:co-diag_opt}, a circuit that simultaneously diagonalize the Pauli products of $\mathcal{S}'$ is produced by Algorithm~\ref{alg:diagonalization} when $\mathcal{S}'$ is given as input.
Thus, if $\begin{bmatrix} \mathcal{S}' & \mathcal{S} \end{bmatrix}$ is given as input to Algorithm~\ref{alg:diagonalization}, then the constructed circuit is a diagonalization network for $\mathcal{S}$ which containing a minimal number of internal Hadamard gates.
An example of the execution of Algorithm~\ref{alg:internal} is given in Figure~\ref{fig:internal_h_example}.

We propose an algorithm, whose pseudo-code is given in Algorithm~\ref{alg:internal}, to solve the Internal-H-Opt problem optimally by finding the Pauli products constituting $\mathcal{S}'$.
Let $J_m$ be an exchange matrix of size $m \times m$ defined as follows:
\begin{equation}
    J_{m_{i,j}} =
    \begin{cases}
    1 & \text{if $i + j = m - 1$},\\
    0 & \text{otherwise}.
    \end{cases}
\end{equation}
As such, the Pauli products encoded by the columns of the matrix $\mathcal{S}J_m$ are then the same as the ones encoded by $\mathcal{S}$ but in reverse order.
The algorithm starts by performing a call to Algorithm~\ref{alg:diagonalization} to obtain a Clifford circuit $C$ that is a diagonalization network for the sequence of Pauli products encoded in $\mathcal{S}J_m$.
Then, a set of stabilizer generators associated with the inverse of $C$ are encoded in the columns of $\mathcal{S}'$ and a second and final call to Algorithm~\ref{alg:diagonalization} is performed where $\begin{bmatrix} \mathcal{S}' & \mathcal{S} \end{bmatrix}$ is given as input.
We prove that the resulting circuit gives an optimal solution to the Internal-H-Opt problem in the next subsection.
When one uses Algorithm~\ref{alg:internal} to perform the re-synthesis of a circuit, as explained in Section~\ref{sec:extension}, the stabilizer generators associated with the final Clifford operator of the input circuit can be append to the final call to Algorithm~\ref{alg:diagonalization} to obtain a full re-synthesis of the circuit containing both a minimal number of Hadamard gates and internal Hadamard gates.

Note that a set of stabilizer generators associated with the inverse of the Clifford circuit $C$ can be computed in $\mathcal{O}(n^2m)$ using the tableau representation as $C$ is composed of $\mathcal{O}(nm)$ gates and a tableau can be updated in $\mathcal{O}(n)$ operations when a Clifford gate is applied.
The complexity of the algorithm then resides in the two calls made to Algorithm~\ref{alg:diagonalization}.
The first call has a complexity of $\mathcal{O}(n^2m)$ as $\mathcal{S}J_m$ is composed of $m$ Pauli products.
For the second call, $n+m$ Pauli products are given as input because a Clifford operator acting on $n$ qubits has $n$ stabilizer generators.
This induces a complexity of $\mathcal{O}(n^2(n+m)) = \mathcal{O}(n^3 + n^2m)$, which corresponds to $\mathcal{O}(n^2m)$ in the typical case where $n \leq m$.
Thus, the overall complexity of Algorithm~\ref{alg:internal} matches the complexity of Algorithm~\ref{alg:diagonalization}.

\begin{algorithm}[t]
    \caption{Diagonalization network synthesis with a minimal number of internal $H$ gates}
    \label{alg:internal}
	\SetAlgoLined
	\SetArgSty{textnormal}
	\SetKwFunction{proc}{InternalHadamardMinimization}
	\SetKwInput{KwInput}{Input}
	\SetKwInput{KwOutput}{Output}
	\KwInput{A sequence $\mathcal{S}$ of $m$ Pauli products of size $n$.}
	\KwOutput{A diagonalization network for $\mathcal{S}$ with a minimal number of internal $H$ gates.}
    $J_m \leftarrow$ exchange matrix of size $m\times m$\\
    $C \leftarrow \texttt{DiagonalizationNetworkSynthesis}(\mathcal{S}J_m)$\\
    $\mathcal{S}' \leftarrow$ stabilizer generators of the inverse of the Clifford operator associated with $C$\\
    \Return $\texttt{DiagonalizationNetworkSynthesis}(\begin{bmatrix} \mathcal{S}' & \mathcal{S} \end{bmatrix} )$
\end{algorithm}

\begin{figure}[t]
    \begin{subfigure}[T]{0.5\textwidth}
    \begin{center}
    \begin{quantikz}[column sep=0.4cm, row sep=0.35cm]
        & \gate{H} & \qw & \ctrl{1} & \gate{S} & \gate{H} & \qw \\
        & \gate{S} & \gate{H} & \targ{} & \qw & \qw & \qw \\
    \end{quantikz}\vspace{-0.5cm}
    \end{center}
    \caption{\centering Circuit produced by Algorithm~\ref{alg:diagonalization} when $\mathcal{S}J_m$ is given as input.}
    \label{subfig:first_circ}
    \end{subfigure}
    \begin{subfigure}[T]{0.5\textwidth}
    \begin{center}
    \begin{quantikz}[column sep=0.4cm, row sep=0.35cm]
        & \gate{S} & \gate{H} & \ctrl{1} & \gate{H} & \ctrl{1} & \gate{S} & \gate{H} & \qw \\
        & \qw & \qw & \targ{} & \qw & \targ{} & \qw & \qw & \qw \\
    \end{quantikz}\vspace{-0.5cm}
    \end{center}
    \caption{\centering Circuit produced by Algorithm~\ref{alg:diagonalization} when $\begin{bmatrix} \mathcal{S}' & \mathcal{S} \end{bmatrix}$ is given as input. The first 4 gates are a diagonalization network for $\mathcal{S}'$, the last 3 gates are a diagonalization for $\tilde{\mathcal{S}}$. The whole circuit solves the Internal-H-Opt problem for $\mathcal{S}$ with $\rank(A^{(\mathcal{S})}) = 1$ internal Hadamard gates.}
    \label{subfig:final_circ}
    \end{subfigure}
    \begin{subfigure}[b]{0.5\textwidth}
    \begin{equation*}
        \arraycolsep=4.0pt
        \mathcal{S} = \begin{bmatrix}\mathcal{Z} \\ \mathcal{X}\end{bmatrix} =
        \left(\begin{array}{ccc}
                1 & 0 & 0 \\
                1 & 1 & 0 \\\hline
                1 & 0 & 1 \\
                0 & 1 & 0 \\
        \end{array}\right),\quad
        A^\mathcal{(S)} = 
        \left(\begin{array}{ccc}
                0 & 1 & 1 \\
                0 & 0 & 0 \\
                0 & 0 & 0 \\
        \end{array}\right)
    \end{equation*}
    \caption{\centering Initial sequence of Pauli products and its commutativity matrix, which are the same as the example in Figure~\ref{fig:swap_example}.}
    \end{subfigure}
    \begin{subfigure}[b]{0.5\textwidth}
    \begin{equation*}
        \arraycolsep=4.0pt
        \mathcal{S}' = 
        \left(\begin{array}{cc}
                1 & 0 \\
                1 & 1 \\\hline
                1 & 1 \\
                0 & 1 \\
        \end{array}\right),\quad
        \tilde{\mathcal{S}} = 
        \left(\begin{array}{ccc}
                0 & 0 & 1 \\
                1 & 1 & 0 \\\hline
                0 & 1 & 1 \\
                0 & 1 & 1 \\
        \end{array}\right)
    \end{equation*}
    \caption{\centering Sequence of Pauli products $\mathcal{S}'$ and $\tilde{\mathcal{S}}$ such that $\mathcal{S}'$ is associated with the stabilizer generators of the Clifford operator $U^\dag$ where $U$ is implemented by the circuit depicted in Subfigure~\ref{subfig:first_circ}, and $\tilde{\mathcal{S}} = U^\dag\mathcal{S}U$.}
    \end{subfigure}
    \caption{Example of an execution of Algorithm~\ref{alg:internal}.
        For a sequence of Pauli products $\mathcal{S}$~(c), the first call to Algorithm~\ref{alg:diagonalization} will produce a circuit~(a) with associated Pauli products $\mathcal{S}'$~(d).
    The algorithm will then output the circuit produced by Algorithm~\ref{alg:diagonalization} when $\begin{bmatrix} \mathcal{S}' & \mathcal{S} \end{bmatrix}$ is given as input~(b).}
    \label{fig:internal_h_example}
\end{figure}

\subsection{Optimality} \label{sec:h_internal_opt}
This subsection is dedicated to the proof of the following theorem, which states the optimality of Algorithm~\ref{alg:internal}.

\begin{theorem}\label{thm:internal_h}
    Let $\mathcal{S}$ be a sequence of $m$ Pauli products, $A$ be its commutativity matrix and let $C$ be the Clifford circuit returned by Algorithm~\ref{alg:internal} when $\mathcal{S}$ is given as input.
    Then $C$ optimally solves the Internal-H-Opt problem with $\rank(A)$ internal Hadamard gates.
\end{theorem}

We first show that the optimal number of internal Hadamard gates is equal to $\rank(A)$.
Our proof rests on the following proposition.

\begin{proposition}\label{prop:commutativity}
    Let $\mathcal{S}$ be a sequence of $m$ Pauli products, $A$ be its commutativity matrix and let $\bs y, \bs y'$ be such that $A\bs y = \bs 0$ and $A\bs y' = \bs 0$.
    Then the Pauli products encoded by $\mathcal{S}\bs y$ and $\mathcal{S}\bs y'$ are commuting.
\end{proposition}

\begin{proof}
Notice that the Pauli product $\mathcal{S}_{:,i}$ commutes with $\mathcal{S}_{:,j}$ if and only if $(A \oplus A^T)_{i, j} = (A \oplus A^T)_{j, i} = 0$.
Then $\mathcal{S}\bs y'$ commutes with $\mathcal{S}_{:,i}$ if and only if $v_i = 0$, where $\bs v = (A \oplus A^T)\bs y'$.
And $\mathcal{S} \bs y'$ commutes with $\mathcal{S} \bs y$ if and only if $\bs y^T \bs v = \bs y^T (A \oplus A^T)\bs y' = 0$.
As $A\bs y = \bs 0$ and $A\bs y' = \bs 0$, we can show that $\bs y^T(A \oplus A^T)\bs y' = \bs y^T A\bs y' \oplus \bs y^T A^T\bs y' = \bs y^T A\bs y' \oplus (A \bs y)^T\bs y' = 0$, which implies that $\mathcal{S}\bs y$ commutes with $\mathcal{S}\bs y'$.\\
\end{proof}

Based on Proposition~\ref{prop:commutativity} we can show that the optimal number of internal Hadamard gates is equal to $\rank(A)$.
Let $\mathcal{S}'$ be a sequence of Pauli products such that the columns of $\mathcal{S}'$ are forming a spanning set of $\{\mathcal{S}\bs y \mid A\bs y = \bs 0, \bs y \in \mathbb{F}_2^m\}$.
It follows that for all $\bs y$ satisfying $A\bs y = \bs 0$ there exists a vector $\bs y'$ such that $\mathcal{S}\bs y = \mathcal{S}' \bs y'$.
Moreover, Proposition~\ref{prop:commutativity} entails that all the Pauli products of $\mathcal{S}'$ are mutually commuting.
Therefore if the Pauli products encoded in $\mathcal{S}'$ were all to be diagonal, then, for all $\bs y$ satisfying $A\bs y = \bs 0$, the Pauli product $\mathcal{S}\bs y$ would be diagonal, i.e. $\mathcal{X}\bs y = \bs 0$.
Let $C'$ be the circuit resulting from the execution of Algorithm~\ref{alg:diagonalization} when $\mathcal{S}'$ is given as input and let $\tilde{\mathcal{S}}$ be the sequence of Pauli products where, for all $i$, the Pauli product encoded by $\tilde{\mathcal{S}}_{:,i}$ is equal to the Pauli product encoded by $\mathcal{S}_{:,i}$ conjugated by the Clifford operator associated with $C'$.
Let $\tilde{M} = \begin{bmatrix} \tilde{\mathcal{X}} \\ A \end{bmatrix}$, for all $\bs y$ satisfying $A\bs y = \bs 0$ we have $\tilde{\mathcal{X}} \bs y = \bs 0$ because $C'$ performs a simultaneous diagonalization on the Pauli products of $\mathcal{S}'$, as stated by Proposition~\ref{prop:co-diag_opt}.
Consequently we have $\tilde{M} \bs y = \bs 0$ for all $\bs y \in \mathrm{nullspace}(A)$ and so $\rank(\tilde{M}) = \rank(A)$.
Then we can use Algorithm~\ref{alg:diagonalization} to produce a Clifford circuit $\tilde{C}$ that is a diagonalization network for $\tilde{\mathcal{S}}$ and such that $h(\tilde{C}) = \rank(\tilde{M}) = \rank(A)$.
It follows that the Clifford circuit $C' :: \tilde{C}$ is a diagonalization network for $\mathcal{S}$ containing $h(\tilde{C}) = \rank(A)$ internal Hadamard gates.

To solve the Internal-H-Opt problem optimally it is then essential to find a spanning set of $\{\mathcal{S}\bs y \mid A\bs y = \bs 0, \bs y \in \mathbb{F}_2^m\}$, which we encode in the columns of $\mathcal{S}'$.
Constructing such a spanning set naively by finding all $\bs y \in \mathbb{F}_2^m$ satisfying $A\bs y = \bs 0$ would imply a complexity of $\mathcal{O}(m^3)$ using a Gaussian elimination procedure, which is more computationally expensive than minimizing the number of Hadamard gates via Algorithm~\ref{alg:diagonalization} in the case where $n < m$.
Fortunately, we can actually rely on Algorithm~\ref{alg:diagonalization} to compute $\mathcal{S}'$ with a complexity of $\mathcal{O}(n^2m)$, as it is done in Algorithm~\ref{alg:internal}.
Indeed, if Algorithm~\ref{alg:diagonalization} is used to constructed a diagonalization network $C$ for the sequence of Pauli products $\mathcal{S}J_m$, then the stabilizer generators of the Clifford operator implemented by $C$ are forming a spanning set of $\{\mathcal{S}\bs y \mid A\bs y = \bs 0, \bs y \in \mathbb{F}_2^m\}$.
We demonstrate this statement via the following proposition.

\begin{proposition}\label{prop:internal_h_opt}
    Let $\mathcal{S} = \begin{bmatrix}\mathcal{Z} \\ \mathcal{X}\end{bmatrix}$ be a sequence of $m$ Pauli products, $J_m$ be an exchange matrix of size $m \times m$ and let $U$ be the Clifford operator associated with the Clifford circuit $C$ produced by Algorithm~\ref{alg:diagonalization} when $\mathcal{S}J_m$ is given as input.
    Let $\tilde{\mathcal{S}}$ be the sequence of Pauli products obtained by conjugating all the Pauli products of $\mathcal{S}$ by $U$, then $\tilde{\mathcal{X}}J_m\bs y = \bs 0$ for all $\bs y$ satisfying $\bs y^T A^{(\mathcal{S}J_m)} = \bs 0$, where $A^{(\mathcal{S}J_m)}$ is the commutativity matrix associated with $\mathcal{S}J_m$.
\end{proposition}

\begin{proof}
Let $C^{(i)}$ be the circuit obtained after the $i$th recursive call to Algorithm~\ref{alg:diagonalization} when $\mathcal{S}J_m$ is given as input, as such $C^{(i)}$ is a diagonalization network for the first $i+1$ columns of $SJ_m$.
And let $\mathcal{S}^{(i)} = \begin{bmatrix}\mathcal{Z}^{(i)} \\ \mathcal{X}^{(i)}\end{bmatrix}$ be the sequence of Pauli products resulting from conjugating $\mathcal{S}$ by the Clifford operator associated with the circuit $C^{(i)}$.
We defined $\bs y^{(i)} \in \mathbb{F}_2^m$ as follows:
\begin{equation}
    y^{(i)}_j =
    \begin{cases}
    y_j & \text{if $j \leq i$},\\
    0 & \text{otherwise},
    \end{cases}
\end{equation}
where $\bs y \in \mathbb{F}_2^m$ satisfies $\bs y^T A^{(\mathcal{S}J_m)} = \bs 0$.

In the case where $i = 0$, the equality $\mathcal{X}^{(0)}J_m\bs y^{(0)} = \bs 0$ is satisfied because the Pauli product encoded by the first column of $\mathcal{S}^{(0)}J_m$ is diagonal and $y^{(0)}_j = 0$ for all $j > 0$.
More generally, the Pauli product encoded by the $i$th column of $\mathcal{S}^{(i)}J_m$ is diagonal, and so the following holds:
\begin{equation}\label{eq:sum_xjm}
    \bigoplus_{k \in K} \mathcal{X}_k^{(i)}J_m = A_{i}^{(\mathcal{S}J_m)} \oplus A_{:,i}^{(\mathcal{S}J_m)}
\end{equation}
where $K = \{k \mid \mathcal{Z}_{k, m-i-1}^{(i)} = 1\}$.
Here the $i$th column of $A^{(\mathcal{S}J_m)}$ must be added to the $i$th row of $A^{(\mathcal{S}J_m)}$ to form the vector describing how the $i$th Pauli rotation commutes or anticommutes with the other Pauli rotations of the sequence.
In this sense, it is a generalization of Equation~\ref{eq:diag_sum} to the other rows of $A^{(\mathcal{S}J_m)}$.
Equation~\ref{eq:sum_xjm} entails
\begin{equation}
    \left[\bigoplus_{k \in K} \mathcal{X}_k^{(i)}J_m\right]^T \bs y^{(i)} = \left[A_{i}^{(\mathcal{S}J_m)} \oplus A_{:,i}^{(\mathcal{S}J_m)}\right]^T \bs y^{(i)}
\end{equation}
Moreover, we have $\left[A_{i}^{(\mathcal{S}J_m)}\right]^T \bs y^{(i)} = 0$ because $A_{i,j}^{(\mathcal{S}J_m)} = 0$ for all $j \le i$ and $y_j^{(i)} = 0$ for all $j > i$.
And we also have $\left[A_{:,i}^{(\mathcal{S}J_m)}\right]^T \bs y^{(i)} = 0$ because $\left[A_{:,i}^{(\mathcal{S}J_m)}\right]^T \bs y^{(i)} = \left[A_{:,i}^{(\mathcal{S}J_m)}\right]^T \bs y$ as $A_{j,i}^{(\mathcal{S}J_m)} = 0$ for all $j > i$ and $\left[A_{:,i}^{(\mathcal{S}J_m)}\right]^T \bs y = \bs y^T A_{:,i}^{(\mathcal{S}J_m)} = 0$ by definition.
Thus, we proved that the following holds: 
\begin{equation}\label{eq:k_sum}
    \left[\bigoplus_{k \in K} \mathcal{X}_k^{(i)}J_m\right]^T \bs y^{(i)} = 0
\end{equation}
where $K = \{k \mid \mathcal{Z}_{k, m-i-1}^{(i)} = 1\}$.

Let's assume that $\mathcal{X}^{(i)} J_m \bs y^{(i)} = \bs 0$, we can then distinguish two cases for the $(i+1)$th iteration of Algorithm~\ref{alg:diagonalization}.
In the case where the $(i+1)$th Pauli product of $\mathcal{S}^{(i)}J_m$ is diagonal, the circuit $C^{(i+1)}$ can be obtained from $C^{(i)}$ by appending a $\{\mathrm{CNOT}, S\}$ circuit to it.
If the Pauli product encoded by $\mathcal{S}^{(i)} J_m \bs y^{(i)}$ is diagonal, as we assumed, then the Pauli product encoded by the vector $\mathcal{S}^{(i+1)} J_m \bs y^{(i)}$ is also diagonal as no Hadamard gate was appended to $C^{(i)}$ to derive $C^{(i+1)}$ from it.
In addition, the $(i+1)$th Pauli product of $\mathcal{S}^{(i+1)}J_m$ is also diagonal which imply that the Pauli product encoded by the vector $\mathcal{X}^{(i+1)} J_m \bs y^{(i+1)}$ is diagonal and so $\mathcal{X}^{(i+1)} J_m \bs y^{(i+1)} = \bs 0$.
Therefore, in such case where the $(i+1)$th Pauli product of $\mathcal{S}^{(i)}J_m$ is diagonal, the equality $\mathcal{X}^{(i)} J_m \bs y^{(i)} = \bs 0$ implies that $\mathcal{X}^{(i+1)} J_m \bs y^{(i+1)} = \bs 0$.

In the case where the $(i+1)$th Pauli product of $\mathcal{S}^{(i)}J_m$ is not diagonal, 
the circuit $C^{(i+1)}$ can be constructed from $C^{(i)}$ by appending a $\{\mathrm{CNOT}, S\}$ circuit to it and a final Hadamard gate on some qubit $j$.
Let $\hat{C}^{(i+1)}$ be the circuit resulting from appending this $\{\mathrm{CNOT}, S\}$ circuit to $C^{(i)}$, i.e.\ $\hat{C}^{(i+1)}$ corresponds to the circuit $C^{(i+1)}$ whose last gate, which is a Hadamard gate, has been removed.
Let $\hat{\mathcal{S}}^{(i+1)}$ be the sequence of Pauli products obtained by conjugating all the Pauli products of $\mathcal{S}$ by the Clifford operator associated with $\hat{C}^{(i+1)}$.
Using the same reasoning as before, if the Pauli product encoded by $\mathcal{S}^{(i)} J_m \bs y^{(i)}$ is diagonal then the Pauli product encoded by the vector $\hat{\mathcal{S}}^{(i+1)} J_m \bs y^{(i)}$ is also diagonal as no Hadamard gate was appended to $C^{(i)}$ to derive $\hat{C}^{(i+1)}$ from it, and so we have $\hat{\mathcal{X}}^{(i+1)} J_m \bs y^{(i)} = \bs 0$.

The circuit $C^{(i+1)}$ can be obtained from $\hat{C}^{(i+1)}$ by appending a Hadamard gate to it on some qubit $j$.
Therefore, $\mathcal{X}_k^{(i+1)} = \hat{\mathcal{X}}_k^{(i+1)}$ for all $k \neq j$, and so
\begin{equation}
    \left[\mathcal{X}_k^{(i+1)}J_m\right]^T \bs y^{(i)} = \left[\hat{\mathcal{X}}_k^{(i+1)}J_m\right]^T \bs y^{(i)} = 0
\end{equation}
for all $k \neq j$.
The $(i+1)$th Pauli product of $\mathcal{S}^{(i+1)}J_m$ is diagonal which means that the $(i+1)$th column of $\mathcal{X}^{(i+1)}J_m$ is equal to $\bs 0$, and so the equality holds as well for $\bs y^{(i+1)}$:
\begin{equation}\label{eq:1}
    \left[\mathcal{X}_k^{(i+1)}J_m\right]^T \bs y^{(i+1)} = \left[\mathcal{X}_k^{(i+1)}J_m\right]^T \bs y^{(i)} = 0
\end{equation}
for all $k \neq j$.
Notice that $j \in K$ where $K = \{k \mid \mathcal{Z}_{k, m-i-1}^{(i+1)} = 1\}$, then from Equation~\ref{eq:k_sum} we can infer that 
\begin{equation}\label{eq:2}
    \left[\mathcal{X}_j^{(i+1)}J_m\right]^T \bs y^{(i+1)} \oplus \left[\bigoplus_{k \in \hat{K}} \mathcal{X}_k^{(i+1)}J_m\right]^T \bs y^{(i+1)} = 0
\end{equation}
where $\hat{K} =  K \setminus \{j\}$.
From Equation~\ref{eq:1} we can deduce that the second term of Equation~\ref{eq:2} is equal to $0$, therefore we have
\begin{equation}
    \left[\mathcal{X}_j^{(i+1)}J_m\right]^T \bs y^{(i+1)} = 0
\end{equation}
which, when combined with Equation~\ref{eq:1}, entails $\mathcal{X}^{(i+1)} J_m \bs y^{(i+1)} = \bs 0$ and concludes the proof of Proposition~\ref{prop:internal_h_opt}.\\
\end{proof}

We can now demonstrate Theorem~\ref{thm:internal_h} on the basis of Proposition~\ref{prop:internal_h_opt}.
\begin{proof}[Proof of Theorem~\ref{thm:internal_h}]
Let $\mathcal{S}'$ be as defined in Algorithm~\ref{alg:internal} and let $C$ be the circuit produced by Algorithm~\ref{alg:internal} when $\mathcal{S} = \begin{bmatrix}\mathcal{Z} \\ \mathcal{X}\end{bmatrix}$ is given as input.
As $C$ is a diagonalization network for $\begin{bmatrix} \mathcal{S}' & \mathcal{S} \end{bmatrix}$ it can be splitted in two subcircuits such that $C = C_1 :: C_2$, where $C_1$ and $C_2$ are diagonalization networks for $\mathcal{S}'$ and $\tilde{\mathcal{S}}$ respectively with $\tilde{\mathcal{S}} = \begin{bmatrix}\tilde{\mathcal{Z}} \\ \tilde{\mathcal{X}}\end{bmatrix} = U^\dag\mathcal{S}U$ where $U$ is the Clifford operator associated with $C_1$.
The number of internal Hadamard gates in $C$ is therefore equal to the number of Hadamard gates in $C_2$, proving Theorem~\ref{thm:internal_h} can then be done by proving that $h(C_2) = \rank(\tilde{M}) = \rank(A^{(\mathcal{S})})$ where $\tilde{M} = \begin{bmatrix} \tilde{\mathcal{X}} \\ A^{(\mathcal{S})} \end{bmatrix}$.

The Pauli products encoded in the matrix $\mathcal{S}J_m$ are the same as in $\mathcal{S}$ but in reverse order.
Consequently we have $A_{i, j}^{(\mathcal{S})} = A_{m-j-1, m-i-1}^{(\mathcal{S}J_m)}$, therefore by reversing the order of the rows and columns of $A^{(\mathcal{S}J_m)}$ and transposing it to obtain a strictly upper triangular matrix we get the matrix $A^{(\mathcal{S})}$:
\begin{equation}
    \left[J_m A^{(\mathcal{S}J_m)} J_m\right]^T = A^{(\mathcal{S})}
\end{equation}
From this we can deduce that
\begin{equation}
\begin{aligned}
    && A^{(\mathcal{S})}\bs y                               &= \bs 0 \\
    \Rightarrow && \left[J_m A^{(\mathcal{S}J_m)} J_m\right]^T\bs y   &= \bs 0 \\
    \Rightarrow && \bs y^T J_m A^{(\mathcal{S}J_m)} J_m               &= \bs 0 \\
    \Rightarrow && \overline{\bs y}^T A^{(\mathcal{S}J_m)}            &= \bs 0
\end{aligned}
\end{equation}
where $\overline{\bs y} = J_m \bs y$.
And based on Proposition~\ref{prop:internal_h_opt} we have
\begin{equation}
\begin{aligned}
                && \tilde{\mathcal{X}}J_m\overline{\bs y}   &= \bs 0 \\
    \Rightarrow && \tilde{\mathcal{X}}\bs y                 &= \bs 0
\end{aligned}
\end{equation}
Thus, for all $\bs y \in \mathrm{nullspace}(A^{(\mathcal{S})})$ we have $\tilde{\mathcal{X}} \bs y = \bs 0$ and therefore $\tilde{M} \bs y = \bs 0$, which implies that $h(C_2) = \rank(\tilde{M}) = \rank(A^{(\mathcal{S})})$ and concludes the proof of Theorem~\ref{thm:internal_h}.\\
\end{proof}

\section{Improving the complexity}\label{sec:improving}

\begin{algorithm}
    \caption{H-Opt}
    \label{alg:circ_h_opt}
	\SetAlgoLined
	\SetArgSty{textnormal}
	\SetKwFunction{proc}{HOpt}
	\SetKwInput{KwInput}{Input}
	\SetKwInput{KwOutput}{Output}
    \KwInput{A Clifford$+R_Z$ circuit $C$ and a tableau $\mathcal{T} = \begin{bmatrix}\bs s^T\\ \mathcal{Z} \\\mathcal{X}\end{bmatrix}$.}
    \KwOutput{A circuit $C_{out}$ and a tableau $\mathcal{T}_{out}$, such that $C_{out}$ is a re-synthesis of $C$ and implements the same sequence of Pauli rotations as $C$ up to an initial and final Clifford operator represented by $\mathcal{T}$ and $\mathcal{T}^{-1}_{out}$ respectively.}
	\SetKwProg{Fn}{procedure}{}{}

    \Fn{\proc{$C$, $\mathcal{T}$}}{
        $C_{out} \leftarrow$ new empty circuit\\
        \ForEach{gate $G \in C$}{
            \If{$G$ is Clifford}{
                Prepend $G^\dag$ to $\mathcal{T}$\\
            }
            \If{$G$ is a non-Clifford $R_{Z_k}(\theta)$ gate}{
                \If{$\exists i$ such that $\mathcal{X}_{i,k} = 1$}{
                    \ForEach{$j \in \{j \mid \mathcal{X}_{j,k} = 1\}\setminus \{i\}$}{
                        $C_{out} \leftarrow C_{out} :: \mathrm{CNOT}_{i, j}$\\
                        Append $\mathrm{CNOT}_{i, j}$ to $\mathcal{T}$\\
                    }
                    \If{$\mathcal{Z}_{i,k} = 1$}{
                        $C_{out} \leftarrow C_{out} :: S_i$\\
                        Append $S_i$ to $\mathcal{T}$\\
                    }
                    $C_{out} \leftarrow C_{out} :: H_i$\\
                    Append $H_i$ to $\mathcal{T}$\\
                }
                $i \leftarrow$ any value satisfying $\mathcal{Z}_{i,k} = 1$\\
                $\tilde{C} \leftarrow$  new empty circuit\\
                \ForEach{$j \in \{j \mid \mathcal{Z}_{j,k} = 1\}\setminus \{i\}$}{
                    $\tilde{C} \leftarrow \tilde{C} :: \mathrm{CNOT}_{j, i}$\\
                }
                \If{$s_k = 1$}{
                    $\tilde{C} \leftarrow \tilde{C} :: X_i$\\
                }
                $C_{out} \leftarrow C_{out} :: \tilde{C} :: R_{Z_i}(\theta) :: \tilde{C}^{-1}$
            }
        }
        \Return $(C_{out}, \mathcal{T})$
	}
\end{algorithm}

\begin{algorithm}
    \caption{Internal-H-Opt}
    \label{alg:circ_internal_h_opt}
	\SetAlgoLined
	\SetArgSty{textnormal}
	\SetKwFunction{proc}{InternalHOpt}
	\SetKwInput{KwInput}{Input}
	\SetKwInput{KwOutput}{Output}
	\KwInput{A Cliffod$+R_Z$ circuit $C$.}
    \KwOutput{A circuit that is a re-synthesis of $C$ and which implements the same sequence of Pauli rotations as $C$ with a minimal number of Hadamard gates and internal Hadamard gates.}
	\SetKwProg{Fn}{procedure}{}{}

	\Fn{\proc{$C$}}{
        $\mathcal{T} \leftarrow$ new identity tableau\\
        \ForEach{Clifford gate $G \in C$}{
            Prepend $G^\dag$ to $\mathcal{T}$\\
        }
        $(\tilde{C}, \tilde{\mathcal{T}}) \leftarrow \texttt{HOpt}(C^{-1}, \mathcal{T})$\\
        $C_{\tilde{\mathcal{T}}} \leftarrow \texttt{CliffordSynthesis}(\tilde{\mathcal{T}})$\\
        $(C_{out}, \mathcal{T}_f) \leftarrow \texttt{HOpt}(C, \tilde{\mathcal{T}})$\\
        \Return $C_{\tilde{\mathcal{T}}} :: C_{out} :: \texttt{CliffordSynthesis}(\mathcal{T}_f^{-1})$
	}
\end{algorithm}

Algorithm~\ref{alg:diagonalization} and~\ref{alg:internal} are taking a sequence of Pauli products $\mathcal{S}$ as input and output a diagonalization network for $\mathcal{S}$.
In order to use these algorithms to minimize the number of Hadamard gates, or internal Hadamard gates, in a circuit $C$ it is then required to first extract from $C$ the sequence of Pauli products $\mathcal{S}$ for which the diagonalization network must be constructed.
This procedure can be done with a complexity $\mathcal{O}(nM)$ by using a tableau, where $n$ is the number of qubits and $M$ is the number of gates in $C$.
In this section, we will see how we can merge the extraction of the sequence of Pauli products $\mathcal{S}$ with our algorithms to obtain the desired re-synthesis of $C$ with a complexity of $\mathcal{O}(nM + n^2h)$ instead of $\mathcal{O}(nM + n^2m)$ where $m$ is the number of Pauli products in $\mathcal{S}$ and $h \leq m$ is the minimal number of Hadamard gates required to construct a diagonalization network for $\mathcal{S}$.
We first explain our notations related to the tableau representation, commonly used to represent a Clifford operator.
In Subsection~\ref{sec:h_opt_alg} we present an algorithm which performs the re-synthesis of a sequence of Pauli rotations implemented by a given circuit up to a final Clifford operator and with a minimal number of Hadamard gates.
Finally, in Subsection~\ref{sec:internal_h_opt_alg}, we present an algorithm which produces a circuit that is a re-synthesis of a given circuit and which implements the same sequence of Pauli rotations but with a minimal number of Hadamard gates and internal Hadamard gates.\\

\noindent\textbf{The tableau representation.}
A tableau encodes $2n$ generators which can be represented by $2n$ independent Pauli products along with a phase for each one of these Pauli products.
We can thus reuse our method of encoding for a sequence of Pauli products $\mathcal{S}$ and represent a tableau by a block matrix $\mathcal{T} = \begin{bmatrix}\bs s^T \\\mathcal{Z} \\ \mathcal{X}\end{bmatrix}$ of size $(2n + 1)\times 2n$ where $n$ is the number of qubits.
The first row of $\mathcal{T}$ corresponds to a vector $\bs s \in \{0, 1\}^{2n}$ which encodes the phases of the generators, the subsequent $n$ rows of $\mathcal{T}$ are forming the submatrix $\mathcal{Z}$ and the last $n$ rows of $\mathcal{T}$ are forming the submatrix $\mathcal{X}$.
The $j$th column of $\mathcal{T}$ is then encoding the $j$th generator: $s_j$ encodes its phase which corresponds to $(-1)^{s_j}$ and $(\mathcal{Z}_{i,j}, \mathcal{X}_{i,j})$ encodes its $i$th Pauli matrix, such that the values $(0, 0), (0, 1), (1, 1)$ and $(1, 0)$ are corresponding to the Pauli matrices $I, X, Y$ and $Z$ respectively.
The first $n$ columns of $\mathcal{T}$ are encoding the stabilizer generators, whereas the last $n$ columns of $\mathcal{T}$ are encoding the destabilizer generators.
The identity tableau $\mathcal{T}$ associated with an empty circuit is such that the matrix $\begin{bmatrix}\mathcal{Z} \\ \mathcal{X}\end{bmatrix}$ is forming the identity matrix and $\bs s = \bs 0$, said otherwise the $i$th stabilizer generator of $\mathcal{T}$ is $Z_i$ and the $i$th destabilizer generator of $\mathcal{T}$ is $X_i$.
The inverse tableau of $\mathcal{T}$, denoted by $\mathcal{T}^{-1}$, is the tableau associated with the Clifford operator $U^\dag$ where $U$ is the Clifford operator associated with $\mathcal{T}$.
Analogously, the inverse of a circuit $C$, denoted $C^{-1}$, is the circuit obtained from $C$ by replacing every gate $G$ by $G^\dag$ and by reversing the order of its gates.
Let $\mathcal{S}$ be a sequence of Pauli products, if $\tilde{\mathcal{S}} = U^\dag \mathcal{S} U$ then we will equivalently say that $\tilde{\mathcal{S}} = \mathcal{T}^{-1}\mathcal{S}\mathcal{T}$ where $\mathcal{T}$ is the tableau associated with the Clifford operator $U$.

Let $C$ be a Clifford circuit such that its associated Clifford operator is represented by a tableau $\mathcal{T}$.
If a Clifford gate from the set $\{\mathrm{CNOT}, S, H\}$ is appended to $C$ then the generators of $\mathcal{T}$ can be updated accordingly with $\mathcal{O}(n)$ operations, where $n$ is the number of qubits.
The operations to perform on the Pauli products encoded by $\mathcal{T}$ are the same as the one depicted in Figure~\ref{fig:clifford_operations}, similar operations can be performed to update the phases associated with the Pauli products in $\mathcal{O}(n)$~\cite{aaronson2004improved}.
Also, if a Clifford gate from the set $\{\mathrm{CNOT}, S, H\}$ is prepended to $C$, then $\mathcal{T}$ can also be updated with $\mathcal{O}(n)$ operations~\cite{gidney2021stim}.
When $\mathcal{T}$ is updated in such manner we will say that we append, or prepend, a gate to $\mathcal{T}$.
As explained in Section~\ref{sec:extension}, a Clifford operator, represented by a tableau $\mathcal{T}$ and acting on $n$ qubits, can be implemented over the $\{X, \mathrm{CNOT}, S, H\}$ gate set with a complexity of $\mathcal{O}(n^3)$ and with a minimal number of Hadamard gates by first diagonalizing its stabilizer generators using Algorithm~\ref{alg:diagonalization}, and then finishing its synthesis using only $\{X, \mathrm{CNOT}, S\}$ gates.
We use the term $\texttt{CliffordSynthesis}$ to denote this procedure in our algorithms.

\subsection{H-Opt algorithm}\label{sec:h_opt_alg}
Consider the algorithm whose pseudo-code is given in Algorithm~\ref{alg:circ_h_opt} and which takes a circuit $C$ and a tableau $\mathcal{T}_{in}$ as input, and let $\mathcal{S}$ be the sequence of Pauli products associated with the sequence of Pauli rotations implemented by $C$.
This algorithm outputs a circuit $C_{out}$ and a tableau $\mathcal{T}$ such that $C_{out}$ is a re-synthesis of $C$ and implements the same sequence of Pauli rotations as $C$ up to an initial and final Clifford operator represented by $\mathcal{T}_{in}$ and $\mathcal{T}^{-1}_{out}$ respectively.

Algorithm~\ref{alg:circ_h_opt} is composed of a loop iterating over the gates of $C$ and which contains two distinct cases: either the current gate $G$ is a Clifford gate or it is not.
If $G$ is Clifford gate then $G^\dag$ is prepended into $\mathcal{T}$.
If $G$ is a non-Clifford $R_{Z_i}(\theta)$ gate then we must compute the Pauli rotation that should be appended to $C_{out}$.
To do so we can first compute which Pauli rotation is actually being implemented by $C$ by pulling all the Clifford gates preceding $G$ through the Pauli rotation $R_{Z_i}(\theta)$.
The Pauli rotation obtained is then $UR_{Z_i}(\theta)U^\dag$ where $U$ is the Clifford operator associated with the Clifford circuit composed of all the Clifford gates preceding $G$.
Then, to be appended into $C_{out}$, the Pauli rotation must also be propagated through the initial tableau $\mathcal{T}_{in}$, we will denote $V$ the Clifford operator associated with $\mathcal{T}_{in}$.
Finally, the Pauli rotation must be propagated through all the Clifford gates that are in $C_{out}$ so far, we denote $W$ the associated Clifford operator.
The Pauli rotation to append to the circuit $C_{out}$ is therefore $W^\dag V^\dag U R_{Z_i}(\theta) U^\dag V W$.
We can notice that the Clifford operator $U^\dag V W$ is in fact associated with the tableau $\mathcal{T}$.
Indeed, $\mathcal{T}$ is initially equal to $\mathcal{T}_{in}$, the inverse Clifford gates that are preceding $G$ in $C$ has been prepended to $\mathcal{T}$ and the Clifford gates that are in $C_{out}$ so far has been appended to $\mathcal{T}$.
The Pauli operator $P$ satisfying $R_P(\theta) = W^\dag V^\dag U R_{Z_i}(\theta) U^\dag V W$ is therefore the $i$th stabilizer generator of $\mathcal{T}$, which is encoded by the $i$th column of $\mathcal{T}$.

The Pauli rotation $R_P(\theta)$ can then be implemented by first performing the synthesis of a Clifford operator $U$ such that $U^\dag R_P(\theta) U$ is diagonal, and then by performing the synthesis of a Clifford operator $V$ satisfying $V^\dag U^\dag R_P(\theta) UV = R_{Z_i}(\theta)$, for some qubit $i$.
The Clifford operator $V$ can be synthesized using only $\{X, \mathrm{CNOT}\}$ gates as $U^\dag R_P(\theta)U$ is diagonal, this is done in Algorithm~\ref{alg:circ_h_opt} by constructing the circuit $\tilde{C}$.
Once the operators $U$ and $V$ have been implemented, the gate $R_{Z_i}(\theta)$ can be appended to the circuit.
The operator $V^\dag$ does not necessarily need to be implemented, but it is actually implemented in Algorithm~\ref{alg:circ_h_opt} to avoid additional operations that would be required to update the tableau $\mathcal{T}$.
We should not treat the operator $U^\dag$ the same way as it would increase the number of Hadamard gates in the circuit, $U^\dag$ is therefore not implemented in Algorithm~\ref{alg:circ_h_opt} and the tableau $\mathcal{T}$ is updated accordingly by appending the gates realizing the implementation of $U$ to it.
Note that the method utilized to implement $U$ is the same as the one in Algorithm~\ref{alg:diagonalization}, which uses exactly one Hadamard gate when $P$ is not diagonal.
It follows from the results in Section~\ref{sec:h_opt} that Algorithm~\ref{alg:circ_h_opt} can be used to solve the H-Opt problem for $\tilde{\mathcal{S}} = \mathcal{T}_{in}^{-1}\mathcal{S}\mathcal{T}_{in}$.
More concretely, by removing all the non-Clifford $R_Z$ gates from the circuit produced by Algorithm~\ref{alg:circ_h_opt} we obtain a diagonalization network which solves the H-Opt problem for $\tilde{\mathcal{S}}$.

Let $C'$ and $C'_{out}$ be the Clifford circuits obtained by removing all the non-Clifford $R_Z$ gates from $C$ and $C_{out}$ respectively, and let $C_{\mathcal{T}_{in}}$ be a Clifford circuit whose Clifford operator is associated with the tableau $\mathcal{T}_{in}$.
In the end of Algorithm~\ref{alg:circ_h_opt}, the tableau $\mathcal{T}$ is associated with the Clifford operator implemented by the circuit $C_{\mathcal{T}} = C'^{-1} :: C_{\mathcal{T}_{in}} :: C'_{out}$.
As $C_{out}$ implements the sequence of Pauli rotations associated with $\tilde{\mathcal{S}} = \mathcal{T}_{in}^{-1}\mathcal{S}\mathcal{T}_{in}$ up to a final Clifford operator implemented by $C'^{-1}_{out}$, it follows that $C_f = C_{\mathcal{T}_{in}} :: C_{out} :: C^{-1}_{\mathcal{T}}$ is a re-synthesis of $C$ and implements the same sequence of Pauli rotations as $C$.
If the input tableau $\mathcal{T}_{in}$ is the identity tableau, or can be implemented with no Hadamard gates, and if $C_{\mathcal{T}}$ is implemented with a minimal number of Hadamard gates using the procedure described in Section~\ref{sec:extension}, then $C_f$ is a re-synthesis of $C$ which implements the same sequence of Pauli rotations with a minimal number of Hadamard gates.\\

\noindent\textbf{Complexity analysis.}
The main loop of Algorithm~\ref{alg:circ_h_opt} is performing $M$ iterations where $M$ is the number of gates in the input circuit.
At each iteration, if the current gate is a Clifford gate then it is prepended to $\mathcal{T}$ which is done in $\mathcal{O}(n)$ operations, where $n$ is the number of qubits in the input circuit.
If the current gate is a non-Clifford $R_{Z_k}(\theta)$ gate then the algorithm append $\mathcal{O}(n)$ gates to $C_{out}$.
In the case where the $k$th stabilizer generator of $\mathcal{T}$ is not diagonal then a subset of these gates are appended to $\mathcal{T}$ which takes $\mathcal{O}(n)$ operations for each gates.
This happens exactly $h$ times where $h$ is the number of Hadamard gates in the output circuit $C_{out}$, which implies a cost of $\mathcal{O}(n^2h)$ operations.
Thus, the overall complexity of Algorithm~\ref{alg:circ_h_opt} is $\mathcal{O}(nM + n^2h)$.

\subsection{Internal-H-Opt algorithm}\label{sec:internal_h_opt_alg}
Algorithm~\ref{alg:circ_internal_h_opt} is based on the procedure explained in Section~\ref{sec:internal} and utilized by Algorithm~\ref{alg:internal} to synthesize a diagonalization network for a sequence of Pauli products with a minimal number of internal Hadamard gates.
It takes a Clifford$+R_Z$ circuit $C$ as input and outputs a circuit which is a re-synthesis of $C$ and which implements the same sequence of Pauli rotations as $C$ with a minimal number of Hadamard gates and internal Hadamard gates.

As explained in Section~\ref{sec:internal}, in order to solve the Internal-H-Opt problem for a sequence of $m$ Pauli products $\mathcal{S}$ it is necessary to find a Clifford operator $U$ that minimizes $\rank(\tilde{M})$ where $\tilde{M} = \begin{bmatrix}\tilde{\mathcal{X}}\\ A^{(\mathcal{S})}\end{bmatrix}$, $A^{(\mathcal{S})}$ is the commutativity matrix associated with $\mathcal{S}$ and $\tilde{\mathcal{S}} = \begin{bmatrix}\tilde{\mathcal{Z}} \\ \tilde{\mathcal{X}}\end{bmatrix} = U^\dag \mathcal{S}U$.
We proved that the Clifford operator associated with the circuit produced by Algorithm~\ref{alg:diagonalization} when $\mathcal{S}J_m$, where $J_m$ is an exchange matrix of size $m\times m$, is given as input is satisfying this property.
Let $\mathcal{S}$ be a sequence of $m$ Pauli products associated with the sequence of Pauli rotations implemented by a Clifford$+R_Z$ circuit $C$, then the Clifford operator $U$ described above can be computed by the \texttt{HOpt} procedure described in Algorithm~\ref{alg:circ_h_opt}.
To do so, the circuit $C^{-1}$ and the tableau $\mathcal{T}$ must be given as input to the \texttt{HOpt} procedure, such that $\mathcal{T}$ is the tableau associated with the Clifford operator implemented by the circuit $C'^{-1}$ where $C'$ is the Clifford circuit obtained by removing all the non-Clifford $R_Z$ gates of $C$.
The circuit $C^{-1}$ is provided so that the Pauli rotations are processed in reversed order by the \texttt{HOpt} procedure.
For the tableau $\mathcal{T}$, it must be provided because the circuit $C^{-1}$ does not necessarily implements the same sequence of Pauli rotations as $C$, however the circuit $C' :: C^{-1}$ do implement the same sequence of Pauli rotations as the circuit $C$.
We can be convinced by this fact by noticing that the Clifford operator formed by all the Clifford gates preceding a non-Clifford gate in $C$ is the same as the Clifford operator formed by all the Clifford gates preceding the corresponding non-Clifford gate in $C' :: C^{-1}$.
Then, as shown in Section~\ref{sec:h_opt_alg}, when the \texttt{HOpt} procedure is executed with $C^{-1}$ and $\mathcal{T}$ as parameters, it will produce a circuit $\tilde{C}$ and a tableau $\tilde{\mathcal{T}}$ associated with the Clifford operator implemented by the circuit $C' :: C'^{-1} :: \tilde{C}'$, which is equivalent to the circuit $\tilde{C}'$, and where $\tilde{C}'$ is the Clifford circuit obtained by removing all the non-Clifford $R_Z$ gates from $\tilde{C}$.
The circuit $\tilde{C}'$ then solves the H-Opt problem for $\mathcal{S}J_m$, and is an implementation of the Clifford operator associated with the tableau $\tilde{\mathcal{T}}$.
From the results of Section~\ref{sec:internal}, it follows that if $\tilde{\mathcal{S}} = \begin{bmatrix}\tilde{\mathcal{Z}} \\ \tilde{\mathcal{X}}\end{bmatrix} = \tilde{\mathcal{T}}^{-1} \mathcal{S}\tilde{\mathcal{T}}$ then $\rank(\tilde{M}) = \rank(A^{(\mathcal{S})})$ where $\tilde{M} = \begin{bmatrix}\tilde{\mathcal{X}}\\ A^{(\mathcal{S})}\end{bmatrix}$.

Algorithm~\ref{alg:circ_internal_h_opt} then performs the synthesis of the Clifford operator associated with $\tilde{\mathcal{T}}$ with a minimal number of Hadamard gates, the Clifford circuit $C_{\tilde{\mathcal{T}}}$ obtained will be the initial Clifford circuit of the circuit produced by Algorithm~\ref{alg:circ_internal_h_opt}.
The algorithm then calls a second time the \texttt{HOpt} procedure with $C$ and $\tilde{\mathcal{T}}$ given as parameters in order to implement the sequence of Pauli rotations associated with $\tilde{\mathcal{S}}$ with a minimal number of internal Hadamard gates and up to a final Clifford circuit.
The tableau $\tilde{\mathcal{T}}$ must be given as input so that the sequence of Pauli rotations implemented is the one associated with the sequence of Pauli products $\tilde{\mathcal{S}}$ and not $\mathcal{S}$.
The procedure \texttt{HOpt} will then produce a circuit $C_{out}$ and a tableau $\mathcal{T}_f$ such that $C'_{out}$ is solving the H-Opt problem for $\tilde{\mathcal{S}}$ and $\mathcal{T}_f$ is associated with the Clifford operator implemented by $C_f = C'^{-1} :: C_{\tilde{\mathcal{T}}} :: C'_{out}$, where $C'$ and $C'_{out}$ are the circuits obtained by removing all the non-Clifford $R_Z$ gates from $C$ and $C_{out}$ respectively.
We can then deduce that $C_{\tilde{\mathcal{T}}} :: C_{out} :: C_f^{-1}$ is implementing the same sequence of Pauli rotations as $C$ and the Clifford operator formed by all the Clifford gates of this circuit is the same as the Clifford operator formed by all the Clifford gates of $C$.
Thus, the circuit produced by Algorithm~\ref{alg:circ_internal_h_opt} is a re-synthesis of $C$ and it implements the same sequence of Pauli rotations with a minimal number of Hadamard gates and internal Hadamard gates.\\

\noindent\textbf{Complexity analysis.}
Let $h$ be the number of Hadamard gates within the circuit produced by Algorithm~\ref{alg:circ_internal_h_opt}, and let $n$ be the number of qubits in $C$.
The algorithm performs two calls to the \texttt{HOpt} procedure for $C^{-1}$ and $C$ respectively, which both contains $M$ gates.
The first call, with $C^{-1}$ given as input, will produce a circuit which contains $\tilde{h}$ number of Hadamard gates, such that $\tilde{h} \leq h$.
The second call, with $C$ given as input, will produce a circuit which contains a number of Hadamard gates that is equal to the number of internal Hadamard gates in the circuit produced by Algorithm~\ref{alg:circ_internal_h_opt}, and which is therefore less than or equal to $h$.
Hence, these two calls to Algorithm~\ref{alg:circ_h_opt} have a cost of $\mathcal{O}(nM + n^2h)$ operations.
The procedure \texttt{CliffordSynthesis} is also called two times, which induces a cost of $\mathcal{O}(n^3)$ operations.
Thus, the overall complexity of Algorithm~\ref{alg:circ_internal_h_opt} is $\mathcal{O}(nM + n^2h + n^3)$, which corresponds to $\mathcal{O}(nM + n^2h)$ in the typical case where $h > n$.

Note that the two calls to the \texttt{CliffordSynthesis} procedure can be avoided if the objective is to minimize the number of internal Hadamard gates in the circuit and not the number of Hadamard gates.
Indeed, the first call to the \texttt{HOpt} procedure will produce a circuit $\tilde{C}$ and a tableau $\mathcal{T}$ such that $\mathcal{T}$ is associated with the Clifford operator implemented by $\tilde{C}'$ where $\tilde{C}'$ is obtained by removing all the non-Clifford $R_Z$ gates from $\tilde{C}$.
Performing the synthesis of $\mathcal{T}$ will therefore produce a circuit that is equivalent to $\tilde{C}'$.
Consequently, instead of calling the procedure \texttt{CliffordSynthesis}, an equivalent circuit could be obtained by constructing $\tilde{C}'$ which can be done with $\mathcal{O}(nM)$ operations as $\tilde{C}$ contains $\mathcal{O}(nM)$ gates.
Of course, the drawbacks of this method are that $\tilde{C}'$ may not contain an optimal number of Hadamard gates and that the worst-case complexity would be greater than $\mathcal{O}(n^3)$ in the case where $M > n^2$.
The second call to \texttt{CliffordSynthesis} can also be avoided in a similar manner.
Indeed, $\mathcal{T}_f$ is associated with the Clifford operator implemented by $C_f = C'^{-1} :: C_{\tilde{\mathcal{T}}} :: C'_{out}$, where $C'$ and $C'_{out}$ are the circuits obtained by removing all the non-Clifford $R_Z$ gates from $C$ and $C_{out}$ respectively.
The circuit $C_f$ can then be constructed in $\mathcal{O}(nM)$ as the circuits $C'^{-1}$, $C_{\tilde{\mathcal{T}}}$ and $C'_{out}$ all contain $\mathcal{O}(nM)$ gates.
Thus, we can design an algorithm which produces a circuit $\hat{C}$ with a complexity of $\mathcal{O}(nM + n^2h)$, even in the case where $h < n$, and such that $\hat{C}$ is a re-synthesis of a Clifford$+R_Z$ circuit $C$ and implements the same sequence of Pauli rotations as $C$ but with a minimal number of internal Hadamard gates.

\section{Benchmarks}\label{sec:bench}

\begin{table}
\resizebox{1.0\columnwidth}{!}{
\begin{tabular}{lrrrrrrrrrrr} 
        \toprule
         & \multicolumn{3}{c}{\texttt{InternalHOpt}} && \multicolumn{3}{c}{\texttt{TMerge}~\cite{zhang2019optimizing} + \texttt{InternalHOpt}} && \multicolumn{3}{c}{\texttt{moveH}~\cite{de2020fast}} \\
        \cmidrule(lr){2-4} \cmidrule(lr){6-8} \cmidrule(lr){10-12}
        Circuit & $H$-count & $T$-count & Time (s) && $H$-count & $T$-count & Time (s) && $H$-count & $T$-count & Time (s) \\
        \midrule
        Tof$_3$ & 2 & 21 & 0.00 && 2 & 15 & 0.00 && 2 & 15 & 0.00 \\ 
        Tof$_4$ & 4 & 35 & 0.00 && 4 & 23 & 0.00 && 4 & 23 & 0.00 \\ 
        Tof$_5$ & 6 & 49 & 0.00 && 6 & 31 & 0.00 && 6 & 31 & 0.00 \\ 
        Tof$_{10}$ & 16 & 119 & 0.00 && 16 & 71 & 0.00 && 16 & 71 & 0.00 \\ 
        Barenco Tof$_3$ & 3 & 28 & 0.00 && 3 & 16 & 0.00 && 3 & 16 & 0.00 \\ 
        Barenco Tof$_4$ & 7 & 56 & 0.00 && 7 & 28 & 0.00 && 7 & 28 & 0.00 \\ 
        Barenco Tof$_5$ & 11 & 84 & 0.00 && 11 & 40 & 0.00 && 11 & 40 & 0.00 \\ 
        Barenco Tof$_{10}$ & 31 & 224 & 0.00 && 31 & 100 & 0.01 && 31 & 100 & 0.00 \\ 
        Mod5$_4$ & 0 & 28 & 0.00 && 0 & 8 & 0.00 && 0 & 8 & 0.00 \\ 
        VBE Adder$_3$ & 4 & 70 & 0.00 && 4 & 24 & 0.00 && 4 & 24 & 0.00 \\ 
        CSLA MUX$_3$ & 6 & 70 & 0.00 && 6 & 62 & 0.00 && 6 & 62 & 0.00 \\ 
        CSUM MUX$_9$ & 12 & 196 & 0.00 && 12 & 84 & 0.01 && 12 & 84 & 0.00 \\ 
        QCLA Com$_7$ & 18 & 203 & 0.00 && 18 & 95 & 0.01 && 18 & 95 & 0.00 \\ 
        QCLA Mod$_7$ & 58 & 413 & 0.00 && 58 & 237 & 0.02 && 58 & 237 & 0.00 \\ 
        QCLA Adder$_{10}$ & 25 & 238 & 0.00 && 25 & 162 & 0.01 && 25 & 162 & 0.00 \\ 
        Adder$_8$ & 41 & 399 & 0.00 && 37 & 173 & 0.02 && 41 & 215 & 0.01 \\ 
        Mod Adder$_{1024}$ & 304 & 1995 & 0.00 && 304 & 1011 & 0.12 && 304 & 1011 & 0.06 \\ 
        RC Adder$_6$ & 10 & 77 & 0.00 && 10 & 47 & 0.00 && 10 & 47 & 0.00 \\ 
        Mod Red$_{21}$ & 17 & 119 & 0.00 && 17 & 73 & 0.00 && 17 & 73 & 0.00 \\ 
        Mod Mult$_{55}$ & 3 & 49 & 0.00 && 3 & 35 & 0.00 && 3 & 35 & 0.00 \\ 
        GF$(2^4)$ Mult & 0 & 112 & 0.00 && 0 & 68 & 0.00 && 0 & 68 & 0.00 \\ 
        GF$(2^5)$ Mult & 0 & 175 & 0.00 && 0 & 115 & 0.01 && 0 & 115 & 0.00 \\ 
        GF$(2^6)$ Mult & 0 & 252 & 0.00 && 0 & 150 & 0.01 && 0 & 150 & 0.00 \\ 
        GF$(2^7)$ Mult & 0 & 343 & 0.00 && 0 & 217 & 0.02 && 0 & 217 & 0.01 \\ 
        GF$(2^8)$ Mult & 0 & 448 & 0.00 && 0 & 264 & 0.04 && 0 & 264 & 0.02 \\ 
        GF$(2^9)$ Mult & 0 & 567 & 0.00 && 0 & 351 & 0.05 && 0 & 351 & 0.03 \\ 
        GF$(2^{10})$ Mult & 0 & 700 & 0.00 && 0 & 410 & 0.07 && 0 & 410 & 0.04 \\ 
        GF$(2^{16})$ Mult & 0 & 1792 & 0.01 && 0 & 1040 & 0.43 && 0 & 1040 & 0.14 \\ 
        GF$(2^{32})$ Mult & 0 & 7168 & 0.05 && 0 & 4128 & 7.19 && 0 & 4128 & 0.98 \\ 
        GF$(2^{64})$ Mult & 0 & 28672 & 0.19 && 0 & 16448 & 125.07 && 0 & 16448 & 7.46 \\ 
        GF$(2^{128})$ Mult & 0 & 114688 & 1.20 && 0 & 65664 & 2294.64 && 0 & 65664 & 60.47 \\ 
        GF$(2^{256})$ Mult & 0 & 458752 & 8.22 && 0 & 262400 & 41474.34 && 0 & 262400 & 2922.20 \\ 
        GF$(2^{512})$ Mult & 0 & 1835008 & 53.85 && - & - & - && 0 & 1049088 & 59186.15 \\ 
        Adder$_{1024}$ & 2044 & 14322 & 3.57 && 2044 & 8184 & 31.08 && 2046 & 8184 & 6.12 \\ 
        Adder$_{2048}$ & 4092 & 28658 & 18.98 && 4092 & 16376 & 179.07 && 4094 & 16376 & 25.69 \\ 
        Adder$_{4096}$ & 8188 & 57330 & 90.46 && 8188 & 32760 & 1182.67 && 8190 & 32760 & 131.11 \\ 
        DEFAULT & 11936 & 62720 & 13.72 && 11936 & 39744 & 39.33 && 12030 & 39744 & 1602.60 \\ 
        Shor$_{4}$ & 9780 & 68320 & 0.21 && 5010 & 17052 & 5.91 && 9829 & 22514 & 77.52 \\ 
        Shor$_{8}$ & 69759 & 489741 & 1.74 && 35585 & 121341 & 158.91 && 69759 & 163827 & 6895.79 \\ 
        Shor$_{16}$ & 537630 & 3755115 & 15.80 && 312274 & 1042881 & 2821.94 && - & - & - \\ 
        Shor$_{32}$ & 4173389 & 29622691 & 172.98 && 387103 & 1303156 & 24150.54 && - & - & - \\ 
        \bottomrule
\end{tabular}}
\caption{Comparison of different methods for the optimization of the number of internal Hadamard gates.
The $H$-count corresponds to the number of internal Hadamard gates.
A blank entry indicates that the execution couldn't be carried out in less than a day.}\label{tab:bench}
\end{table}

We compare the performances of Algorithm~\ref{alg:circ_internal_h_opt}, the \texttt{InternalHOpt} procedure, to the \texttt{moveH} procedure presented in Reference~\cite{de2020fast} and which has a complexity of $\mathcal{O}(M^2)$ where $M$ is the number of gates in the input circuit.
Note that the \texttt{moveH} procedure does not include the $T$ gates reduction method of Reference~\cite{de2020fast} based on spider nest identities and which is normally performed once the number of Hadamard gates have been reduced.
The \texttt{moveH} procedure applies a sequence of rewriting rules on the circuit with the aim of reducing the number of internal Hadamard gates.
During this process the number of $T$ gates may also be reduced, which modifies the sequence of Pauli rotations implemented by the circuit.
This can then lead to a better reduction in the number of internal Hadamard gates than the one obtained when only the \texttt{InternalHOpt} procedure is performed.
Which is why, in order to better exploit the \texttt{InternalHOpt} procedure, it can be helpful to first execute an algorithm which can reduce the number of $T$ gates in the circuit quickly and efficiently.
The $T$-count reduction algorithms that are closest to these requirements are the provided in Reference~\cite{zhang2019optimizing} and in Reference~\cite{kissinger2020reducing}, these two algorithms have in fact been proven to be equivalent~\cite{simmons2021relating}.
The method used in these algorithms consists in merging the Pauli rotations in the sequence that are equivalent and that are not separated by another Pauli rotation with which they anticommute.
We implemented the algorithm provided in Reference~\cite{zhang2019optimizing} such that it is not increasing the number of gates in the circuit in order to not increase the execution time of the \texttt{InternalHOpt} procedure.
This procedure, which we refer to as \texttt{TMerge}, has a complexity of $\mathcal{O}(nM + nm^2)$ where $n$ is the number of qubits, $M$ is the number of gates in the input circuit and $m$ is the number of Pauli rotations.
If the $T$ gates reduction rules used in the \texttt{moveH} subroutine is only consisting in merging two adjacent $R_Z$ gates together, then we can infer that the number of $T$ gates in the circuit after \texttt{moveH} procedure has been performed is always higher or equal to the number of $T$ gates in the circuit after the \texttt{TMerge} procedure has been performed; this is corroborated by the results of our benchmarks.

We evaluate the different methods on a set of commonly used circuits which were obtained from Reference~\cite{amyGithub} and Reference~\cite{reversibleBenchmarks}. 
We extended the set of circuits over which the benchmarks are performed by adding larger quantum circuit to better test the scalability of the different approach on various types of circuits.
We added large adders circuits which are performing an addition over two registers of size $1024$, $2048$ and $4096$ qubits, the implementation of these circuits is based on Reference~\cite{takahashi2010quantum}.
We also added a circuit, given in Reference~\cite{default}, that is an implementation of the block cipher DEFAULT.
Finally, we added quantum circuits implementing the modular exponentiation part of Shor's algorithm for number factoring over 4, 8, 16 and 32 bits.

The \texttt{TMerge} and \texttt{InternalHOpt} procedures were implemented with the Rust programming language, while the \texttt{moveH} procedure was extracted from the implementation realized in Haskell by the authors of the method~\cite{stompCode}.
Our implementation of the \texttt{InternalHOpt} procedure used for the benchmarks is publicly available~\cite{github}, along with the circuits used in the benchmarks and which have a reasonable size.
The operations performed by the \texttt{InternalHOpt} algorithm mostly consist in bitwise operations between vectors in order to update the tableau.
Thus, our algorithm can greatly benefits from SIMD (Same Instruction Multiple Data) instructions which enable the simultaneous execution of some of these bitwise operations.
This have for example been used in the CHP stabilizer circuit simulator~\cite{aaronson2004improved}.
We also exploit this concept in our implementation of the \texttt{InternalHOpt} procedure by using 256 bit wide Advanced Vector Extensions (AVX).\\

\noindent\textbf{Benchmarks analysis.}
The results of our benchmarks are presented in Table~\ref{tab:bench}.
We can notice that the \texttt{InternalHOpt} procedure outperforms the \texttt{moveH} procedure in term of execution time on some circuits of large size.
For instance, the Shor$_{32}$ circuit was optimized in $173$ seconds by the \texttt{InternalHOpt} procedure while the two other methods did not succeed in optimizing the circuit in less than a day.
However, the \texttt{InternalHOpt} procedure alone does not always achieve the best results in the number of internal Hadamard gates.
For the set of circuits and methods considered, the method achieving the best results in term of internal Hadamard gates is the $\texttt{TMerge} + \texttt{InternalHOpt}$ approach.
Indeed, the $\texttt{TMerge} + \texttt{InternalHOpt}$ approach always leads to a number of internal Hadamard gates that is lower or equal to the numbers obtained by the \texttt{moveH} procedure.
This fact also holds for the number of $T$ gates.
However, for some circuits, the performances of the \texttt{moveH} procedure and the $\texttt{TMerge} + \texttt{InternalHOpt}$ approach are similar with respect to the $H$-count and $T$-count metrics, but the execution time of the \texttt{moveH} procedure is much lower.
This is notably the case for the adder circuits of large size.
These adder circuits have a low depth and a high number of qubits, which is far from the ideal case for $\texttt{TMerge} + \texttt{InternalHOpt}$ approach since the complexity of both procedures is dependent on the number of qubits.
On the contrary, the \texttt{moveH} procedure is not affected by the number of qubits as it has a complexity of $\mathcal{O}(M^2)$ where $M$ is the number of gates within the circuit.
This explains why the \texttt{moveH} procedure is competitive for these adder circuits and has an execution time that is close to the one of the \texttt{InternalHOpt} procedure.

Another series of circuits for which the \texttt{moveH} procedure is much faster than the $\texttt{TMerge} + \texttt{InternalHOpt}$ approach are the ``GF($2^n$) Mult'' circuits.
This behaviour can be explained by analyzing the structure of the ``GF($2^n$) Mult'' circuits and the design of the \texttt{TMerge} algorithm.
The ``GF($2^n$) Mult'' circuits are all implementing a sequence of Pauli rotations that are mutually commuting, which is why no internal Hadamard gate is required for these circuits.
In the worst case, for every pair of Pauli rotations, the \texttt{TMerge} procedure will check whether two Pauli rotations commute or not.
This routine, which seems unnecessary in the case where we know that the Pauli rotations are all mutually commuting, is particularly expensive for the ``GF($2^n$) Mult'' circuits for which $n$ is high since the number of Pauli rotations increases drastically with respect to $n$.\\

\noindent\textbf{Outlook.}
Our primary motivation for optimizing the number of internal Hadamard gate is to foster the minimization of $T$-gates.
Conversely, our benchmarks show that optimizing the number of $T$-gate leads to better minimization in the number of internal Hadamard gates.
This interdependence between the $T$-count and $H$-count minimization problems could lead us to think that a second round of $T$-count optimization followed by a $H$-count optimization could lead to a lower number of internal Hadamard gates.
Our investigations on that second round of optimization have not be fruitful as we did not succeed in reducing the number of internal Hadamard gates below the numbers obtained by the $\texttt{TMerge} + \texttt{InternalHOpt}$ approach.
It seems that once the \texttt{TMerge} procedure has been performed, it becomes difficult to modify the underlying sequence of Pauli rotations in such a way that it enables further reduction in the number of internal Hadamard gates.
Our conclusion here is only based on some of our tests, more investigations with a wide variety of $T$-count optimizers should be performed to know whether or not this second round of optimization could lead to an improvement in the number of internal Hadamard gates.

Two lines of investigations on how to perform the optimization of internal Hadamard gates more efficiently can be drawn out from these benchmarks.
Firstly, the \texttt{TMerge} procedure is outperformed, with respect to the execution time, by the \texttt{moveH} procedure on some circuits such as the ``GF($2^n$) Mult'' circuits, can the complexity of the \texttt{TMerge} procedure be improved so that it is more competitive on these circuits?
Secondly, is it possible to design an algorithm similar to the \texttt{moveH} procedure, so that it has approximatively the same execution time, but which systematically obtains the same number of $T$ gates as the \texttt{TMerge} procedure and which optimally minimizes the number of internal Hadamard gates in the resulting sequence of Pauli rotations as done by the \texttt{InternalHOpt} procedure?

\section{Conclusion}
We presented an algorithm to realize the synthesis of a sequence of Pauli rotations over the $\{X, \mathrm{CNOT}, S, H, R_Z\}$ gate set using a minimal number of Hadamard gates and with a time complexity of $\mathcal{O}(n^2m)$, in the typical case where $n \leq m$, and where $n$ is the number of qubits and $m$ is the number of Pauli rotations.
A closely related problem is to optimize a Clifford$+R_Z$ circuit so that the sequence of Pauli rotations it is implementing contains a minimal number of internal Hadamard gates, where a Hadamard gate is called internal if it is comprised between the first and last non-Clifford $R_Z$ gates of the circuit.
Solving this problem is important to improve the efficiency and scalability of algorithms minimizing the number of non-Clifford $R_Z$ gates such as $T$-count optimizers, and to minimize the additional cost that comes with the Hadamard gates gadgetization procedure.
In Reference~\cite{de2020fast}, the authors raised the question of whether this problem is solvable in $\mathcal{O}(M^2 \text{poly}\log(M))$ time where $M$ is the number of gates in the input circuit.
We answer this question positively, in the case where $n \leq M/\sqrt{h}$ and for a fixed sequence of Pauli rotations by providing an algorithm solving this problem with a time complexity of $\mathcal{O}(nM + n^2h)$ where $n$ is the number of qubits, $M$ is the number of gates in the input circuit and $h$ is the number of Hadamard gates within the optimized circuit.

Our algorithms are optimal for a given sequence of Pauli rotations, however there may exist other sequences of Pauli rotations, associated with the same operator, which could be implemented with fewer Hadamard gates.
An open problem is to find a sequence of Pauli rotations $\mathcal{S}$ implementing a given unitary gate up to a Clifford operator such that $\rank(A^{(\mathcal{S})})$ is minimal, where $A^{(\mathcal{S})}$ is the commutativity matrix associated with $\mathcal{S}$ as defined in Section~\ref{sec:h_opt}.
Should there exist an algorithm solving this problem in reasonable time, then it could be used in conjunction with our algorithms to implement a unitary gate over the Clifford$+R_Z$ gate set with a minimal number of internal Hadamard gates.

\section*{Acknowledgments}
We acknowledge funding from the Plan France 2030 through the projects NISQ2LSQ ANR-22-PETQ-0006 and EPIQ ANR-22-PETQ-007. 

\bibliographystyle{unsrt}
\bibliography{ref.bib}

\appendix

\section{Diagonalization network synthesis example}\label{app:example}

In this section we provide a detailed execution example of Algorithm~\ref{alg:diagonalization} which performs the synthesis of a diagonalization network for a given sequence of Pauli products.
Let $\mathcal{S}$ be the sequence of Pauli products given as input to Algorithm~\ref{alg:diagonalization} and defined as follows:

\begin{equation*}
    \arraycolsep=4.0pt
    \mathcal{S} = \begin{bmatrix}\mathcal{Z} \\ \mathcal{X}\end{bmatrix} = 
    \left(\begin{array}{cccc}
            1 & 1 & 1 & 0 \\
            0 & 1 & 0 & 1 \\
            0 & 1 & 0 & 0 \\\hline
            1 & 0 & 1 & 1 \\
            0 & 1 & 1 & 0 \\
            0 & 0 & 1 & 1 \\
    \end{array}\right).
\end{equation*}

The algorithm starts by diagonalizing the Pauli product represented by the first column of $\mathcal{S}$.
This is done by inserting a $S$ gate in the circuit followed by a Hadamard gate on the first qubit.
The matrix $\mathcal{S}$ encoding the sequence of Pauli products is updated by performing the operations associated with the $S$ and $H$ gates, as depicted in Figure~\ref{fig:clifford_operations}.
Then, the first column is removed from the matrix and the algorithm performs a recursive call on the updated matrix.

\begin{figure}[h]
\begin{subfigure}{0.69\textwidth}
\centering
\begin{quantikz}[column sep=0.35cm, row sep=0.5cm]
    \qw & \gate{S} & \gate{H} & \qw \\
    \qw & \qw & \qw & \qw \\
    \qw & \qw & \qw & \qw \\
\end{quantikz}
\end{subfigure}
\begin{subfigure}{0.3\textwidth}
\begin{equation*}
    \arraycolsep=4.0pt
    \left(\begin{array}{cccc}
            0 & 1 & 1 \\
            1 & 0 & 1 \\
            1 & 0 & 0 \\\hline
            1 & 0 & 1 \\
            1 & 1 & 0 \\
            0 & 1 & 1 \\
    \end{array}\right)
\end{equation*}
\end{subfigure}
\end{figure}

This time, the first column of the lower matrix has a Hamming weight greater than one.
Therefore, the algorithm inserts a CNOT gate acting on the first and second qubits of the circuit to reduce the Hamming weight of the first column of the lower matrix to one.
The Pauli product encoded by the first column can then be diagonalized by inserting a $S$ gate and a Hadamard gate on the first qubit.
Then, the matrix is updated, the first column is removed from the matrix and the algorithm performs a recursive call.

\begin{figure}[h]
\begin{subfigure}{0.69\textwidth}
\centering
\begin{quantikz}[column sep=0.35cm, row sep=0.5cm]
    \qw & \gate{S} & \gate{H} & \ctrl{1} & \gate{S} & \gate{H} & \qw \\
    \qw & \qw & \qw & \targ{} & \qw & \qw & \qw \\
    \qw & \qw & \qw & \qw & \qw & \qw & \qw \\
\end{quantikz}
\end{subfigure}
\begin{subfigure}{0.3\textwidth}
\begin{equation*}
    \arraycolsep=4.0pt
    \left(\begin{array}{cccc}
            0 & 1 \\
            0 & 1 \\
            0 & 0 \\\hline
            1 & 1 \\
            1 & 1 \\
            1 & 1 \\
    \end{array}\right)
\end{equation*}
\end{subfigure}
\end{figure}

Again, the first column of the lower matrix has a Hamming weight greater than one.
This time two CNOT gates must be inserted in the circuit to reduce it to one.
After that, a Hadamard gate is inserted to diagonalize the Pauli product encoded by the first column.\\~\\

\begin{figure}[h]
\begin{subfigure}{0.69\textwidth}
\centering
\begin{quantikz}[column sep=0.35cm, row sep=0.5cm]
    \qw & \gate{S} & \gate{H} & \ctrl{1} & \gate{S} & \gate{H} & \ctrl{1} & \ctrl{2} & \gate{H} & \qw \\
    \qw & \qw & \qw & \targ{} & \qw & \qw & \targ{} & \qw & \qw & \qw \\
    \qw & \qw & \qw & \qw & \qw & \qw & \qw & \targ{} & \qw & \qw \\
\end{quantikz}
\end{subfigure}
\begin{subfigure}{0.3\textwidth}
\begin{equation*}
    \arraycolsep=4.0pt
    \left(\begin{array}{cccc}
            1 \\
            1 \\
            0 \\\hline
            0 \\
            0 \\
            0 \\
    \end{array}\right)
\end{equation*}
\end{subfigure}
\end{figure}

Finally, the Pauli product encoded by the remaining column is already diagonal.
Therefore, the algorithm simply removes the column from the matrix.
The matrix is then empty so the algorithm terminates by returning the constructed circuit, which is a diagonalization network for the sequence of Pauli products encoded by $\mathcal{S}$.

We can then insert $\{\mathrm{CNOT}, R_Z\}$ subcircuits in the appropriate places to implement the sequence of Pauli rotations associated with $\mathcal{S}$ up to a final Clifford circuit.
The following figure shows an example of a possible circuit obtained after this procedure.

\begin{figure}[h]
    \centering
\begin{quantikz}[column sep=0.35cm, row sep=0.5cm]
    \qw & \gate{S} & \gate{H} & \gate{R_Z} & \ctrl{1} & \gate{S} & \gate{H} & \targ{} & \targ{} & \gate{R_Z} & \targ{} & \targ{} & \ctrl{1} & \ctrl{2} & \gate{H} & \gate{R_Z} & \targ{} & \gate{R_Z} & \targ{} & \qw \\
    \qw & \qw & \qw & \qw & \targ{} & \qw & \qw & \ctrl{-1} & \qw & \qw & \qw & \ctrl{-1} & \targ{} & \qw & \qw & \qw & \ctrl{-1} & \qw & \ctrl{-1} & \qw \\
    \qw & \qw & \qw & \qw & \qw & \qw & \qw & \qw & \ctrl{-2} & \qw & \ctrl{-2} & \qw & \qw & \targ{} & \qw & \qw & \qw & \qw
    & \qw & \qw \\
\end{quantikz}
\end{figure}

The commutativity matrix $A^{(\mathcal{S})}$ associated with $\mathcal{S}$ is 
\begin{equation*}
    A^{(\mathcal{S})} = 
    \left(\begin{array}{cccc}
            0 & 1 & 0 & 1 \\
            0 & 0 & 1 & 1 \\
            0 & 0 & 0 & 0 \\
            0 & 0 & 0 & 0 \\
    \end{array}\right).
\end{equation*}

As stated by Theorem~\ref{thm:alg_h}, we can notice that the number of Hadamard gates in the circuit produced by Algorithm~\ref{alg:diagonalization} is equal to $\rank(M)$ where $M = \begin{bmatrix} \mathcal{X} \\ A^{(\mathcal{S})} \end{bmatrix}$.

\end{document}